\documentclass[12pt]{article}
\usepackage{fullpage}
\usepackage{graphicx,amssymb,amsmath, amsthm}
\usepackage{subfigure}
\usepackage{lineno}

\usepackage{color}
\newcommand{\remove}[1]{{}}
\newcommand{\ignore}[1]{{}}
\newcommand{\rednote}[1]{{\color{red} #1}}
\newcommand{\bluenote}[1]{{\color{blue} #1}}
\newcommand{\changed}[1]{#1}
\definecolor{light-gray}{gray}{0.7}

\newtheorem{lemma}{Lemma}
\newtheorem{theorem}{Theorem}

\title{Visibility Graphs, Dismantlability, and the Cops and Robbers Game}

\author{Anna Lubiw\thanks{Cheriton School of Computer Science,
        University of Waterloo, Waterloo ON, N2L 3G1, Canada.  {\tt \{alubiw, hvosough\}@uwaterloo.ca}}
        \and
        Jack Snoeyink\thanks{Department of Computer Science, University of North Carolina, NC, USA. {\tt snoeyink@cs.unc.edu}}
        \and
        Hamideh Vosoughpour$^*$}

\begin{document}

\maketitle

\begin{abstract}  
% JSS: I would prefer to lead with the definitions, rather than having them after the claims, but I'll leave the previous abstract in comments to make it easy to go back to that. 
We study versions of cop and robber pursuit-evasion games on the visibility graphs of polygons, and inside polygons with straight and curved sides.  Each player has full information about the other player's location, players take turns, and the robber is captured when the cop arrives at the same point as the robber.  
In visibility graphs  we show the cop can always win because visibility graphs are {\it dismantlable},  which is interesting as one of the few results relating visibility graphs to other known graph classes.  
%By results of Nowakowski and WinklerThis implies that the cop wins the cops and robbers game on the visibility graph of a polygon. 
We extend this to show that the cop wins games in which players move along straight line segments inside any  polygon
and, more generally, inside any simply connected planar region with a reasonable boundary. Essentially, our problem is a type of pursuit-evasion using the link metric rather than the Euclidean metric, and our result provides an interesting class of infinite cop-win graphs.
%
%We show that visibility graphs of polygons are dismantlable, which is interesting because there are very few results relating visibility graphs to other known graph classes.
%%By results of Nowakowski and Winkler
%This implies that the cop wins the cops and robbers game on the visibility graph of a polygon. 
%We extend this to show that the cop wins the cops and robbers game inside any  polygon
%and more generally inside any simply connected planar region with a reasonable boundary.
%In this game, a cop and robber take turns moving inside the region. 
%Each player has full information about the other player's location, and the robber is captured when the cop arrives at the same point as the robber.
%In one turn, a player may move any distance along a straight line segment inside the region.  
%This result provides an interesting class of infinite cop-win graphs.
%We discuss relationships  
%to pursuit-evasion in polygonal regions.  Essentially, our problem is a type of pursuit-evasion using the link metric rather than the Euclidean metric.
\end{abstract}

%%%%%%%%%%%%%%%%%%%%%%%%%%%%%%%%%%%%%%%%%%%%%%%%%%%%%%%
\section{Introduction}

Pursuit-evasion games have a rich history both for their mathematical interest and because of applications in surveillance, search-and-rescue, and mobile robotics. 
 In  pursuit-evasion games one player,  called the ``evader,'' tries to avoid capture by ``pursuers'' as all players move in some domain.
There are many game versions, depending on whether the domain is discrete or continuous, what information the players have, and how the players move---taking turns, moving with bounded speed, etc. 
 
This paper is about the ``cops and robbers game,'' a discrete version played on a graph, that was
first introduced in 1983 by
Nowakowski and Winkler~\cite{Nowakowski-83}, and Quilliot~\cite{Quilliot-83}.  
The cop and robber are located at vertices of a graph and take turns moving along edges of the graph.  
 The robber is caught when a cop moves to the vertex the robber is on.  
The standard assumption is that both players have full information about the graph and the other player's location.
The first papers on this game~\cite{Nowakowski-83,Quilliot-83} characterized the graphs in which the cop wins---they are the graphs with a  ``dismantlable'' vertex ordering. (Complete definitions are in Section~\ref{sec:defs}.)
Since then many extensions have been explored---see the book by Bonato and  Nowakowski~\cite{Bonato-book}.   (Note that the  cops and robbers version that Seymour and Thomas~\cite{seymour} develop to characterize treewidth is different:   the robber moves only along edges but at arbitrarily high speed, while a cop may jump to any graph vertex.)

\medskip\noindent
{\bf Our Results.}
We consider three successively more general versions of the cops and robbers game in planar regions.  % and show that all of them are cop-win.
The first version is the cops and robbers game on the {\it visibility graph of a polygon}, which is a graph with a vertex for each polygon vertex, and an edge when two vertices ``see'' each other (i.e., may be joined by a line segment) in the polygon.  
We prove that this game is cop-win by proving that visibility graphs are dismantlable. As explained below, this result is implicit in~\cite{Aichholzer-convexifying-11}.   
We prove the stronger result 
that visibility graphs are 2-dismantlable. 
We remark that it is an open problem to characterize or efficiently recognize visibility graphs of polygons~\cite{ghosh2007visibility, ghosh2010Unsolved}, so this result is significant in that it places visibility graphs as a subset of a known and well-studied class of graphs.

Our second setting is the cops and robbers game on all points inside a polygon.
The cop chooses a point inside the polygon as its initial position, then the robber chooses its initial position.  Then the players take turns, beginning with the cop.  In each turn, a player may move to any point visible from its current location, i.e.,~it may move any distance along a straight-line segment inside the polygon.  The cop wins when it moves to the robber's position. 
We prove that the cop will win using the simple strategy of always taking the first step of a shortest path to the robber.  Thus the cop plays on the reflex vertices of the polygon.
%We give an alternative proof that shows that every polygon can be partitioned into triangles that have a dismantlable ordering (exact definition in Section~\ref{sec:polygon}).  
%We are thus able for the first time to apply the theory of discrete cops and robbers games, namely dismantlability, to a continuous geometric setting.

Our third setting is the cops and robbers game on all points inside a bounded simply-connected planar region.  We show that if the boundary is well-behaved (see below) then the cop wins.   We give a strategy for the cop to win, although the cop can no longer follow the same shortest path strategy (e.g.~when it lies on a reflex curve), and can no longer win by playing on the boundary.

The cops and robbers game on all points inside a region can be viewed as a cops and robbers
game on an infinite graph---the graph has a vertex for each point inside the region, and an edge when two points see each other.   
These ``point visibility graphs''  were introduced by Shermer (see~\cite{Shermer}) for the case of polygons.
% note: recent different use of the term by Ghosh and Roy http://arxiv.org/pdf/1209.2308.pdf
% for visibility of points (not in general position) with no edges involved
Our result shows that point visibility graphs are cop-win.
This provides an answer to Hahn's question~\cite{Hahn-survey} of finding an interesting class of infinite cop-win graphs.

The cops and robbers game on all points inside a region can be viewed as pursuit-evasion under a different metric, and could appropriately be called ``straight-line pursuit-evasion.'' 
Previous work~\cite{Isler-randomized-05,Bhadauria-12}  considered a pursuit-evasion game in a polygon (or polygonal region) where the players
are limited to moving distance 1 in the Euclidean metric on each turn.  In our game, the players are limited to distance 1 in the \emph{link metric}, where the length of a path is number of line segments in the path.
This models a situation where changing direction is  costly but straight-line motion is easy. 
Mechanical robots cannot make instantaneous sharp turns %~\cite{LaValle} \mynote{find ref.}, 
so exploring a model where all turns are expensive is a good first step towards a more realistic analysis of pursuit-evasion games with turn constraints.
We also note that the protocol of the players taking alternate turns is more natural in the link metric than in the Euclidean metric.
%very natural in the link metric.
%If the players can move freely but with limited speed (the Euclidean metric) then there is no natural reason why the evader should remain stationary while the pursuer moves.  However, in a situation where changing direction takes some time---the same amount of time for both players---but straight-line motion is instantaneous, the turn-taking model arises naturally.

%Our paper is organized as follows: \rednote{fill in}.

\section{Related Work}
\label{sec:related-work}

\noindent{\bf Cops and Robbers.}
The cops and robbers game was introduced by Nowakowski and Winkler~\cite{Nowakowski-83}, and Quilliot~\cite{Quilliot-83}.
They characterized  the  finite graphs where one cop can capture the robber (``cop-win'' graphs) as ``dismantlable'' graphs, which can be recognized efficiently.
They also studied infinite cop-win graphs.  
Aigner and Fromme~\cite{Aigner-84} introduced the  \emph{cop number} of a graph, the minimum number of cops needed to catch a robber.  The rule with multiple cops is that they all move at once.
Among other things Aigner and Fromme proved that three cops are always sufficient and sometimes necessary for planar graphs.
Beveridge et al.~\cite{Beveridge-geometric-12} studied \emph{geometric graphs} (where vertices are points in the plane and an edge joins points that are within distance 1) and show that 9 cops suffice, and 3 are sometimes necessary. 
Meyniel conjectured that  $O(\sqrt n)$ cops can catch a robber in any graph on $n$ vertices~\cite{Baird-Meyniel-12}.
For any fixed $k$ there is a polynomial time algorithm to test if
$k$ cops can catch a robber in a given graph, but the problem is NP-complete for general $k$~\cite{Fomin-tractability-08}, and EXPTIME-complete for directed graphs~\cite{Goldstein-95}. 
The cops and robbers game on infinite graphs was 
studied in the original paper~\cite{Nowakowski-83} and others, e.g.~\cite{Bonato-capture-09}.

In a cop-win graph with $n$ vertices, the cop can win in at most $n$ moves.
This result is implicit in the original papers, but a clear exposition can be found in the book of Bonato and Nowakowski~\cite[Section 2.2]{Bonato-book}.
For a fixed number of cops, the number of cop moves needed to capture a robber in a given graph can be computed in 
polynomial time~\cite{Hahn-MacGillivray-06}, but the problem becomes NP-hard in general~\cite{Bonato-capture-09}.
%Many other variations of the  cops and robbers game have been studied~\cite{Isler-hunter-06,Frieze-variations-12,Isler-info-08,Chalopin-hide-11,Fomin-norecharge-10,Bhattacharya-grid-10}.

%For graphs Isler and Karnad\cite{Isler-info-08} explore the case where the cop only knows robber's location when the robber is within some bounded distance from the cop.
%The cop can still catch the robber with a randomized strategy, but the number of moves increases.  
%J.~Chalopin et.~al. in~\cite{Chalopin-hide-11} studied the game when the robber can hide and ride and is visible every $k$ step. They characterize the cop-win graphs for any value of $k$. V.~Isler and N.~Karnad~\cite{Karnad-info-08} considered the role of information for the game of cop and robbers and showed that the game is still cop-win even if the cop does not have any information about the robber's position, but the capture time will be exponential in that case.

\noindent{\bf Pursuit-Evasion.}
In the cops and robbers game, space is discrete.  %In the realm of continuous spaces
For continuous spaces, a main focus has been on polygonal regions, i.e.,~a region bounded by a polygon with polygonal holes removed.
%Early works such as the original 1992 paper by Suzuki and Yamashita~\cite{Suzuki-searching-92} and 
The seminal 1999 paper by Guibas et al.~\cite{Guibas-visibility-99} 
concentrated on ``visibility-based'' pursuit-evasion
where the evader is arbitrarily fast and the pursuers do not know the evader's location and must search the region until they make line-of-sight contact.   This  models the scenario of agents searching the floor-plan of a building to find a smart, fast intruder that can be zapped from a distance.
%capture photographic evidence against a smart, fast intruder.
Guibas et al.~\cite{Guibas-visibility-99} showed that $\Theta(\log n)$ pursuers are needed in a simple polygyon, and more generally they bounded  the number of pursuers in terms of the number of holes in the region.   
%A very recent paper~\cite{Klein-12} shows that these bounds actually depend on the feature size of the polygonal region.  -- this correction is about CAPTURE, not detection, so the Guibas paper is ok.
If the pursuers have the power to make random choices, Isler et al.~\cite{Isler-randomized-05} showed that only one guard  is needed for a polygon.
For a survey on pursuit-evasion in polygonal regions, see~\cite{Chung-pursuit-survey-11}.

The two games (cops and robbers/visibility-based pursuit-evasion) make opposite assumptions on five criteria: 
space is discrete/continuous; 
the pursuers succeed by capture/line-of-sight;
the pursuers have full information/no information; 
the evader's speed is limited/ unlimited; 
time is  discrete/continuous (i.e.,~the players take turns/move continuously).   

The difference between players taking turns and moving continuously can be vital,
 as revealed in Rado's Lion-and-Man problem from the 1930's (see Littlewood~\cite{Littlewood}) where the two players are inside a circular arena and move with equal speed.  The lion wins in the turn-taking protocol, but---surprisingly---the man can escape capture if both players move continuously. 
 
Bhaduaria et al.~\cite{Bhadauria-12} consider a pursuit-evasion game using a model very similar to ours.  
Each player knows the others' positions
%The players know each other's positions 
(perhaps from a surveillance network) and the goal is to actually capture the evader.  Players have equal speed and take turns. 
In a polygonal region they show that 3 pursuers can capture an evader in $O(n d^2)$ moves where $n$ is the number of vertices and $d$ is the diameter of the polygon.  They also give an example where 3 pursuers are needed.
In a simple polygon they show that 1 pursuer can capture an evader in $O(n d^2)$ moves.  
% ~\cite{Isler-randomized-05},
This result, like ours, can be viewed as a result about a cop and robber game on an infinite graph.
The graph in this case has a vertex for each point in a polygon, and an edge 
joining any pair of points at distance at most 1 in the polygon.
%when two points are distance at most 1 apart in the polygon.  
To the best of our knowledge, the connection between this result and  cops and robbers on (finite) geometric graphs~\cite{Beveridge-geometric-12} has not been explored.

There is also a vast literature on graph-based pursuit-evasion games, where players move continuously and have no knowledge of other players' positions.  
The terms  ``graph searching'' and ``graph sweeping'' are used, and the concept is related to tree-width.  For surveys see~\cite{Alspach-04,Fomin-bibliography}.

\noindent{\bf Curved Regions.}
Traditional algorithms in computational geometry deal with points and piecewise linear subspaces (lines, segments, polygons, etc.). 
The study of algorithms for curved inputs was initiated by Dobkin and Souvaine~\cite{souvaine1990}, who defined the widely-used splinegon model.
A \emph{splinegon} is a simply connected region
 formed by replacing each edge of a simple polygon by a curve of constant complexity such that the area bounded by the curve and the edge it replaces is convex.  
The standard assumption is that it takes constant time to perform
primitive operations such as finding the intersection of a line with a splinegon edge or computing common tangents of two splinegon edges. 
This model is widely used as the standard model for curved planar environments in different studies.

Melissaratos and Souvaine~\cite{melissaratos1992} gave a linear time algorithm to find a shortest path between two points in a splinegon.  Their algorithm is similar to shortest path finding in a simple polygon but uses a trapezoid decomposition in place of polygon triangulation.
For finding shortest paths among curved obstacles (the splinegon version of a polygonal domain) there is recent work~\cite{Chen:2013:CSP}, and also more efficient algorithms when the curves are more specialized~\cite{pseudodisks2013,Hershberger:2013}.

\remove{
\bluenote{A draft for here: H.V.}
Traditional algorithms in computational geometry deal with points and piece wise linear subspaces  (lines, segments, polygons, etc.). For real applications we need to handle the environments with curved boundaries. However, the complexity of the curved boundary highly effects the complexity of algorithms. There are some studies those try to solve the problems on curves. In \cite{pseudodisks2013, disks1988} the problem of finding shortest path on special curves like disks and pseudodisks is studied. Some heuristic approaches approximate the curved boundaries by polygons and use the algorithms designed for polygonal domains \cite{Hershberger:2013}. The study of algorithms for more general curved inputs initiated by Dobkin and Souvaine in 1990 \cite{souvaine1990}. They modeled curved boundaries with a finite set of convex splines and call them ``splinegons''. Each spline, like lines, is in $O(1)$ complexity, so primitive operations such as finding the intersection of a line with a spline or the common tangents of two splines can be computed in constant time. This model is widely used as a standard model for curved planar environments in different studies. 

Melissaratos and Souvaine developed a linear time algorithm to find shortest path between two points in a splinegon \cite{melissaratos1992}. They introduced a trapezoid decomposition of the curved region that is a substitute to polygon triangulation in $O(n)$ time. Using the trapezoid decomposition the shortest path in a splinegon can be computed in $O(n)$ time with an algorithm similar to shortest path finding in a simple polygon. 
D.Z.~Chen and H.~Wang \cite{Chen:2013:CSP} developed an $O(n + h\log^{1+\epsilon} h + k)$ algorithm to find shortest path in a splinegon with $h$ splines and $n$ vertices, where $k$ is the number of all common tangents inside the region and is in $O(h^2)$. The algorithm works  for any arbitrary small $\epsilon$. They also proposed an algorithm in $O(n + k + h\log h)$ time to find all common tangents in a splinegon.
}

%%%%%%%
%%%%%%%%%%%%%%%%%%%%%%%%%%%%%%%%%%%%%%%%%%%%
%%%%%%%%%%%%%%%%%%%%%%%%%%%%%%%%%%%%%%
\section{Preliminaries}\label{sec:defs}

For a vertex $v$ of a graph,  we use $N[v]$ to denote the \emph{closed neighbourhood} of $v$, which consists of $v$ together with the vertices adjacent to $v$.  Vertex $v$ \emph{dominates} vertex $u$ if $N[v] \supseteq N[u]$. 

A graph $G$ is \emph{dismantlable} if it has a vertex ordering $\{v_1, v_2, \dots, v_n\}$ such that for each $i<n$, there is a vertex $v_j$, $j>i$  that dominates $v_i$ in the graph $G_i$ induced by $\{v_{i},\dots, v_n\}$.

We regard a polygon as a closed set of points, the interior plus the boundary.
Two points in a polygon
are \emph{visible} or \emph{see} each other if the line segment between them lies inside the polygon.
The line segment %between two visible points 
may lie partially or totally on the boundary of the polygon.
The \emph{visibility graph} of a polygon has the same vertex set as the polygon and an edge between any pair of vertices that see each other in the polygon.
For any point $x$ in polygon $P$, the \emph{visibility polygon} of $x$, $V(x)$, is the set of points 
in $P$ visible from $x$.
Note that $V(x)$ may fail to be a simple polygon---it may have 1-dimensional features on its boundary in certain cases where $x$ lies on a line through a pair of vertices.

\ignore{%%%%%%%%%%   We don't have room for this figure
Fig.~\ref{fig:visibilityPolygon} shows two examples of the visibility polygon of point $x$.

\begin{figure}[h]
	\centering
	\subfloat[\label{subfig:visibilityPolygon:a}]{%
		\includegraphics[width=.3\textwidth]{fig-1-1.pdf}}

	\quad
	\subfloat[\label{subfig:visibilityPolygon:b}]{%
		\includegraphics[width=.3\textwidth]{fig-1-2.pdf}}
	
		\caption{The visibility polygon of $x$.}
	\label{fig:visibilityPolygon}
\end{figure}
} %%%%%%%%%%%%%   end ignore

%This set of points, which are all visible to $x$, always forms a connected region, but as shown in  Fig.~\ref{subfig:visibilityPolygon:b} it does not necessarily form a simple polygon without co-linear vertices. The visibility polygon may also have zero angles.  

For points $a$ and $b$ in polygon $P$, we say that $a$ \emph{dominates} $b$ if $V(a) \supseteq V(b)$.   
Note that we are using ``dominates'' both for vertices in a graph (w.r.t.~neighbourhood containment) and for points in a polygon (w.r.t.~visibility polygon containment).    
For vertices $a$ and $b$ of a polygon, if $a$ dominates $b$ in the polygon then $a$ dominates $b$ in the visibility graph of the polygon, but not conversely.

%%%%%%%%%%%%%%%%%%%%%%%%%%%%%%%%%%%%%%%%%%%%%%%%%%%%%%%%%%%%
\section{Cops and Robbers in Visibility Graphs}
\label{sec:visibility-graph}

In this section we show that the visibility graph of any polygon is cop-win by showing that any such graph is dismantlable.

This result is actually implicit in the work of Aichholzer et al.~\cite{Aichholzer-convexifying-11}.
They defined an
edge $uv$ of polygon $P$ to be \emph{visibility increasing} if
for every two points $p_1$ and $p_2$ in order along the edge $uv$
the visibility polyons nest: $V(p_1) \subseteq V(p_2)$.
In particular, this implies that $v$ dominates every point on the edge, and that $v$ dominates $u$ in the visibility graph.
Aichholzer et~al.\ showed that every polygon has a visibility-increasing edge.  
%(See Lemma~\ref{lemma:maxpocket} below.)
%We will first prove 
It is straight-forward to show 
that visibility graphs are dismantlable based on this result.
% that every polygon has a visibility-increasing edge.
\begin{lemma}  
\label{lemma:dismantlable}
The visibility graph $G$ of any polygon $P$ is dismantlable.
%Assume that every polygon has a visibility-increasing edge.  Then for any polygon $P$ its visibility graph $G$ is dismantlable.
\end{lemma}
\begin{proof}
By induction on the number of vertices of the polygon.  
Let $uv$ be a  visibility-increasing edge, which we know exists by the result of Aichholzer et al. 
Then vertex $v$ dominates $u$ in the visibility graph $G$.  We will construct a dismantlable ordering starting with vertex $u$.  

It suffices to show that $G - u$ is dismantlable. Let $tu$ and $uv$ be
the two polygon edges incident on $u$. 
We claim that the triangle $tuv$  is contained in the
polygon:  $u$ sees $t$ on the polygon boundary, so $v$ must also see
$t$. (Triangle $tuv$ is an ``ear'' of the polygon.)
Removing triangle $tuv$ yields a smaller polygon whose visibility graph is $G-u$.  By induction, $G-u$ is dismantlable.
\end{proof} 

Aichholzer et al.~\cite{Aichholzer-convexifying-11} conjectured that a polygon always has at least two visibility-increasing edges.  
In the remainder of this section we prove this conjecture, thus giving a simpler proof of their result and also proving 
that visibility graphs of polygons are \emph{2-dismantlable}.  Bonato et al.~\cite{Bonato-capture-09} define a graph $G$ to be 2-\emph{dismantlable} if it either has fewer than 7 vertices and is cop-win or it has at least two vertices $a$ and $b$ such that each one is dominated by a vertex other than $a,b$, and such that $G-\{a,b\}$ is 2-dismantlable.
They show that if an $n$-vertex graph is 2-dismantlable then the cop wins in at most $\frac{n}{2}$ moves by choosing the right starting point.

We need a few more definitions.
Let $P$ be a simple polygon, with an edge $uv$
where $v$ is a reflex vertex.
Extend the directed ray from $u$ through $v$ and let $t$ be the first boundary point of $P$ beyond $v$ that the ray hits.
The points $v$ and $t$ divide the boundary of $P$ into two paths. Let $\sigma$ be the path that does not contain $u$.
The simple polygon formed by $\sigma$ plus the edge $vt$ is called a \emph{pocket} and denoted Pocket$(u,v)$.
The segment $vt$ is the \emph{mouth} of the pocket.  Note that $u$ does not see any points inside Pocket$(u,v)$ except points on the line that contains the mouth.
See Fig.~\ref{fig:1} for examples,  including some with collinear vertices, which will arise in our proof.
Pocket($u,v)$ is \emph{maximal} if no other pocket properly contains it.
Note that a non-convex polygon has at least one pocket, and therefore at least one maximal pocket.  This will be strengthened to two maximal pockets in Lemma~\ref{lemma:pockets} below.  

\begin{figure}[h]
\centering
\includegraphics[width=.6\textwidth]{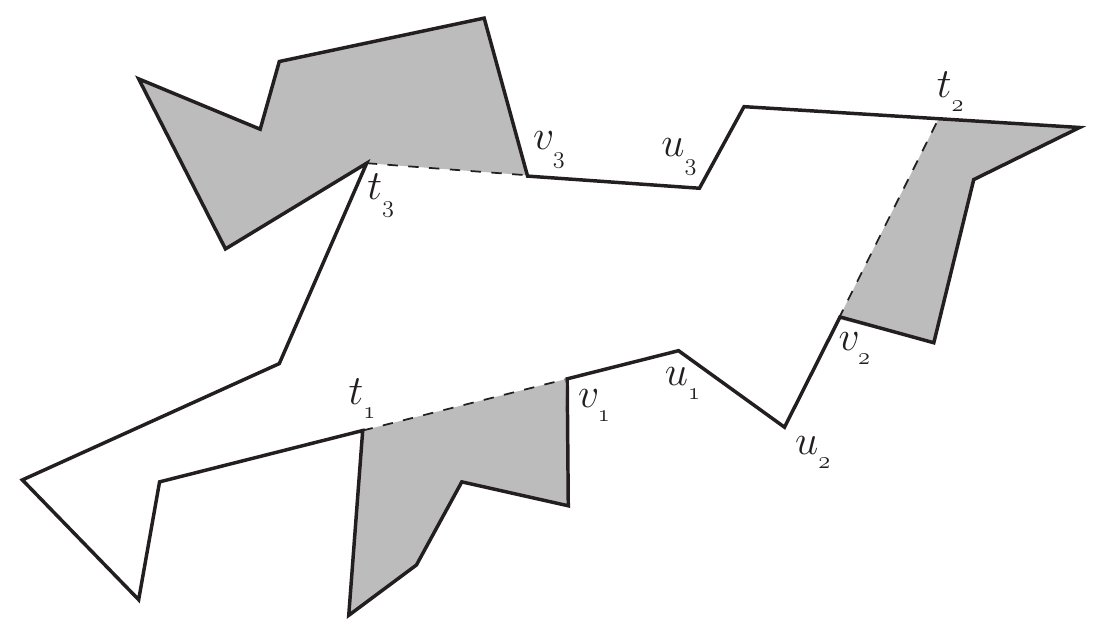}
\caption{Pocket$(u_i,v_i), i=1,2,3$, shaded.}
\label{fig:1}
\end{figure}

To prove that the visibility graph of a  polygon  is 2-dismantlable we prove that  a maximal pocket in the polygon provides a visibility-increasing edge and that every nonconvex polygon has at least two maximal pockets.

%\begin{lemma} 
%If $u,v$ is a maximal pocket then $uv$ is a visibility-increasing edge.
%\end{lemma}
%\begin{proof}
%\end{proof}

\begin{lemma}
\label{lemma:maxpocket}
If $uv$ is an edge of a polygon and Pocket$(u,v)$ is maximal then $uv$ is a visibility-increasing edge.
%In any simple polygon $P$, for the pair of consecutive vertices $u$ and $v$ which forms a maximal $uv$-pocket, $uv$ is a visibility increasing edge. 
\end{lemma}

Aichholzer et al.~\cite[Lemma 2]{Aichholzer-convexifying-11} essentially proved this although it was not expressed in terms of maximal pockets.  
Also they assumed the polygon has no three collinear vertices.  We include a proof.

\begin{proof}  %[of Lemma~\ref{lemma:maxpocket}] 
We prove the contrapositive.
%By contradiction, 
Suppose that edge $uv$ is not visibility-increasing.
If $u$ is a reflex vertex with next neighbour $w$, say, then Pocket$(w,u)$ properly contains Pocket$(u,v)$, which implies that Pocket$(u,v)$ is not maximal.  Thus we may assume that $u$ is convex.
Since $uv$ is not visibility-increasing there are two points $p_1$ and $p_2$ in order along $uv$ such that 
%$p_1$ is closer to $u$ but its 
the visibility polygon of $p_1$ is not contained in the visibility polygon of $p_2$. 
%It means that 
Thus there is a point $t$ which is visible to $p_1$ but not visible to $p_2$.
See Fig.~\ref{subfig:maxpocket:a}.   
By extending the segment $p_1 t$, we may assume, without loss of generality, that $t$ is on the polygon boundary.
We claim that $t$ lies in the closed half-plane bounded by the line through $uv$ and lying on the opposite side of Pocket$(u,v)$.
This is obvious if  $p_1$ is internal to edge $uv$, and if $p_1=u$ it follows because $u$ is convex.
Furthermore, $t$ cannot lie on the line through $u,v$ otherwise $p_2$ would see $t$.

%This is due to a part of polygon which blocks the ray $tp_2$. 
%Now rotate the ray $tp_2$ toward $tp_1$ until the blocking part is passed. 
Now move point $p$ from $p_1$ to $p_2$ stopping at the last point where $p$ sees $t$.
%Suppose that the ray hits $uv$ in point $p$. 
See Fig.~\ref{subfig:maxpocket:b}.
There must be a reflex vertex $v'$ on the segment $tp$. 
%Visibility can only be blocked by a reflex vertex $v'$ on the segment $tp$. 
The points $v'$ and $t$ divide the polygon boundary into two paths.  Take the path that does not contain $v$, and let $u'$ be the first neighbour of $v'$ along this path.  It may happen that $u' = t$.  
%Clearly there is at least one reflex vertex $v'$ on the ray $tp$, 
%The $u'v'$-pocket 
Then, as shown in Fig.~\ref{subfig:maxpocket:b}, Pocket$(u'v')$ %fully 
properly contains %the $uv$-pocket. 
Pocket$(u,v)$, so Pocket$(u,v)$ is not maximal.
%This contradicts the assumption that the $uv$-pocket is maximal.

\begin{figure}[ht]
	\centering
	{\hfill\subfigure[]{\includegraphics[width=.3\textwidth]{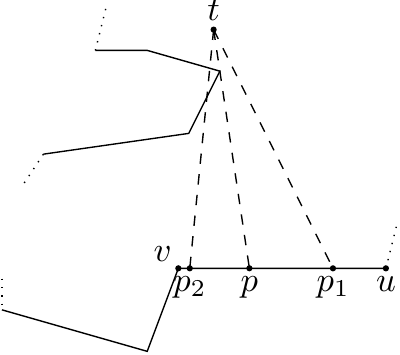}\label{subfig:maxpocket:a}}\hfill
	\subfigure[]{\includegraphics[width=.3\textwidth]{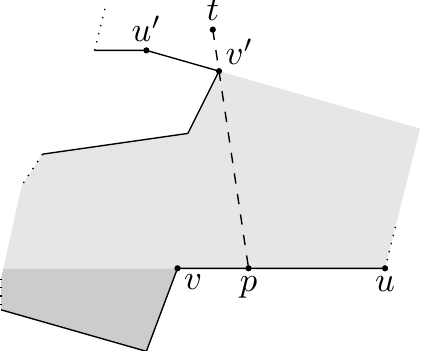}\label{subfig:maxpocket:b}}\hfill}
%		\subfloat[\label{subfig:maxpocket:a}]{%
%		\includegraphics[width=.2\textwidth]{Figures/fig-2-1.pdf}}
%	\quad\quad\quad
%	\subfloat[\label{subfig:maxpocket:b}]{%
%		\includegraphics[width=.2\textwidth]{Figures/fig-2-2.pdf}}
	\caption{If $uv$ is not visibility-increasing then Pocket$(u,v)$ is not maximal.}
	%Not being visibility increasing edge contradicts to a maximal $uv$-pocket}
	\label{fig:maxpocket}
\end{figure}

\end{proof}

\begin{lemma}
\label{lemma:pockets}
Any polygon that is not convex has two maximal pockets Pocket($u_1, v_1)$ and Pocket$(u_2, v_2)$ 
%with $u_1 \ne u_2$ and $u_1$ not adjacent to $u_2$ on the polygon boundary.
where $u_1$ does not see $u_2$.
\end{lemma}
\begin{proof}
Let Pocket$(u_1,v_1)$ be a maximal pocket.  Let $u$ be the other neighbour of $v_1$ on the polygon boundary.  Consider Pocket$(u,v_1)$, which must be contained in some maximal pocket, Pocket$(u_2,v_2)$.  Vertex $u_1$ is inside Pocket$(u,v_1)$ and not on the line of its mouth.
Therefore $u_1$ is inside Pocket$(u_2,v_2)$ and not on the line of its mouth.  Since $u_2$ cannot see points inside Pocket$(u_2,v_2)$ except on the line of its mouth  therefore $u_2$ cannot see $u_1$.  
%\mynote{Figure?}
\end{proof}

From the above lemmas, together with the observation that the
visibility graph of a convex polygon is a complete graph, which is
2-dismantlable, we obtain the result that visibility graphs are 2-dismantlable.
\begin{theorem} The visibility graph of a polygon is 2-dismantlable.
\label{theorem:vis-graph}
\end{theorem}
Consequently, the cop wins the cops and robbers game on the visibility graph of an $n$-vertex polygon in at most $\frac{n}{2}$ steps.  
\changed{There is a lower bound of $n/4$ cop moves in the worst case, as shown by the skinny zig-zag polygon illustrated in Figure~\ref{fig:zig-zag}.  In Section~\ref{sec:lower-bound} we will prove an $\Omega(n)$ lower bound on the number of cop moves even when the cop can move on points interior to the polygon, and even when the polygon has link diameter 3, i.e.~any two points in the polygon are joined by a path of at most 3 segments.

\begin{figure}[ht]
\centering
\includegraphics[width=.6\textwidth]{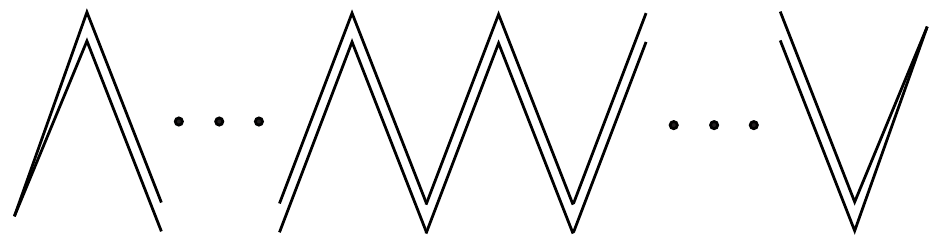}
\caption{An $n$-vertex polygon in which the robber can survive for $n/4$ steps.}
\label{fig:zig-zag}
\end{figure}

}

%%%%%%%%%%%%%%%%%%%%%%%%%%%%%%%%%%%%%%%%%%
%%%%%%%%%%%%%%%%%%%%%%%%%%%%%%%%%%%%%%%%%%
\section{Cops and Robbers Inside a Polygon}
\label{sec:polygon}

In this section we look at the cops and robbers game on all points
inside a polygon.  This is a cops and robbers game on an infinite
graph so induction on dismantlable orderings does not immediately apply.  Instead we give a direct geometric proof that the cop always wins.
Although the next section proves more generally that the cop
always wins in any simply connected planar region with a reasonable
boundary, it is worth first seeing the simpler proof for the polygonal
case, both to gain understanding and because this case has a tight
$\Theta(n)$ bound on the maximum number of moves (discussed in Section~\ref{sec:lower-bound}).

\begin{theorem} 
\label{thm:polygon}
The cop wins the cops and robbers game on the points inside any polygon in at most $n$ steps using the strategy of always taking the first segment of the shortest path from its current position to the robber.
\end{theorem}

\begin{proof}
We argue that each move of the cop restricts the robber to an ever
shrinking {\it active region} of the polygon.  Suppose the cop is
initially at $c_0$ and the robber initially at $r_0$.  In the $i^{\rm
  th}$ move the cop moves to $c_i$ and then the robber moves to $r_i$.

Observe that for $i\ge 1$ points $c_i$ are at reflex vertices of the
polygon.
To define the active region $P_i$ containing the robber position $r_i$, 
we first define its boundary, a directed line segment,
$\bar \ell_i$.  Suppose that the shortest path from $c_{i-1}$ to $r_{i-1}$
turns left at $c_i$, as in Figure~\ref{fig:ray}.  
Define $\bar \ell_i$ to be the segment that starts at $c_i$ and goes
through $c_{i-1}$ and stops at the first boundary point of the polygon
where an edge of the polygon goes to the left  of
the ray $c_i c_{i-1}$.  (If the shortest path turns right at $c_i$ we
similarly define  $\bar \ell_i$ to stop where a polygon edge goes right.)
In general, the segment $\bar \ell_i$ cuts the polygon into two (or more)
pieces; let {\it active region} $P_i$  be the piece that contains
$r_{i-1}$.  (In the very first step, $\bar \ell_1$ may hug the polygon
boundary, so $P_1$ may be all of $P$.)

\begin{figure}[ht]
\centering
\includegraphics[width=.65\textwidth]{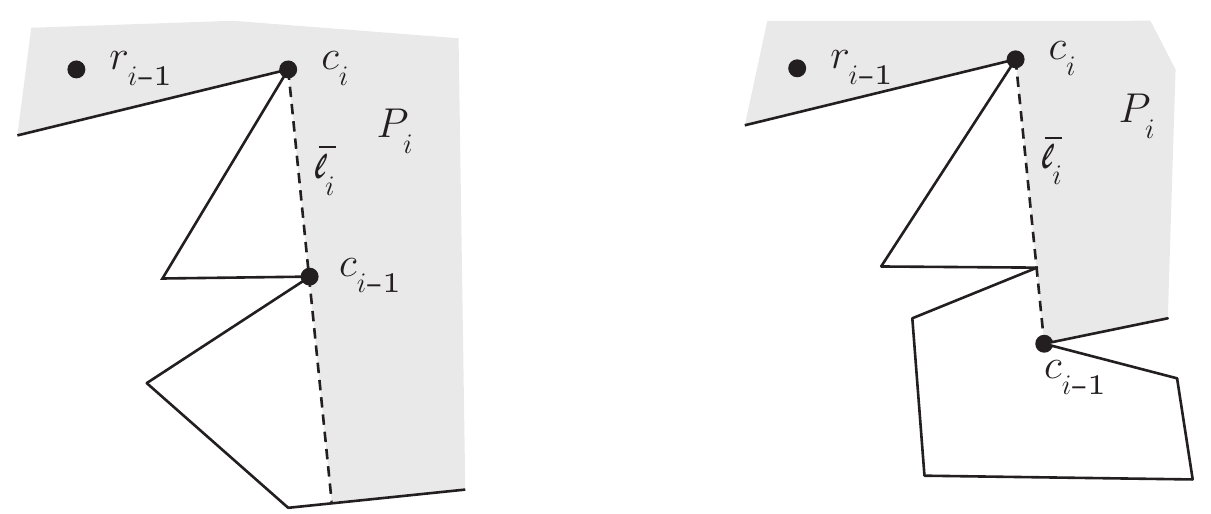}
\caption{The segment $\bar \ell_i$ and the active region $P_i$ (shaded)
  containing robber positions $r_k$, for all $k\ge i-1$.}
\label{fig:ray}
\end{figure}

We claim by induction on the (decreasing) number of vertices of $P_i$ that the robber can never leave $P_i$, i.e.,~that $r_i$, $r_{i+1}$,\dots\ are in $P_i$.  It suffices to show that $r_i$ is in $P_i$ and that $P_{i+1} \subseteq P_i$ and that $P_{i+1}$ has fewer vertices.

%? Renumber so you're looking at $c_0, c_1, c_2$?  %%JSS: No, since
%the initial point c0 is not necessarily reflex
Suppose that the shortest path from $c_{i-1}$ to $r_{i-1}$ turns left at $c_{i}$.
(The other case is completely symmetric.)
Observe that the next robber position $r_i$ must be inside $P_i$, i.e.,~the robber cannot move from $r_{i-1}$ to cross $\bar \ell_i$.
We  distinguish two cases depending whether 
the shortest path from $c_i$ to $r_i$ makes a left or a right turn at $c_{i+1}$. 

\smallskip\noindent{\bf Case 1.} See Figure~\ref{fig:polygon-case1}.  The shortest path from $c_i$ to $r_i$ makes a left turn at $c_{i+1}$. 
We consider two subcases: (a) $c_{i+1}$ is left of the ray $c_{i-1} c_i$; and (b) $c_{i+1}$ is right of (or on) the ray $c_{i-1} c_i$.
We claim that case (b) cannot happen because the robber could not have moved from $r_{i-1}$ to $r_i$---see Figure~\ref{fig:polygon-case1}(b). 
For case (a) observe that  $\bar \ell_{i+1}$ extends past $c_i$ and therefore $P_{i+1}$ is a subset of $P_i$ and smaller by at least one vertex---see Figure~\ref{fig:polygon-case1}(a).

\begin{figure}[ht]
\centering
\includegraphics[width=.65\textwidth]{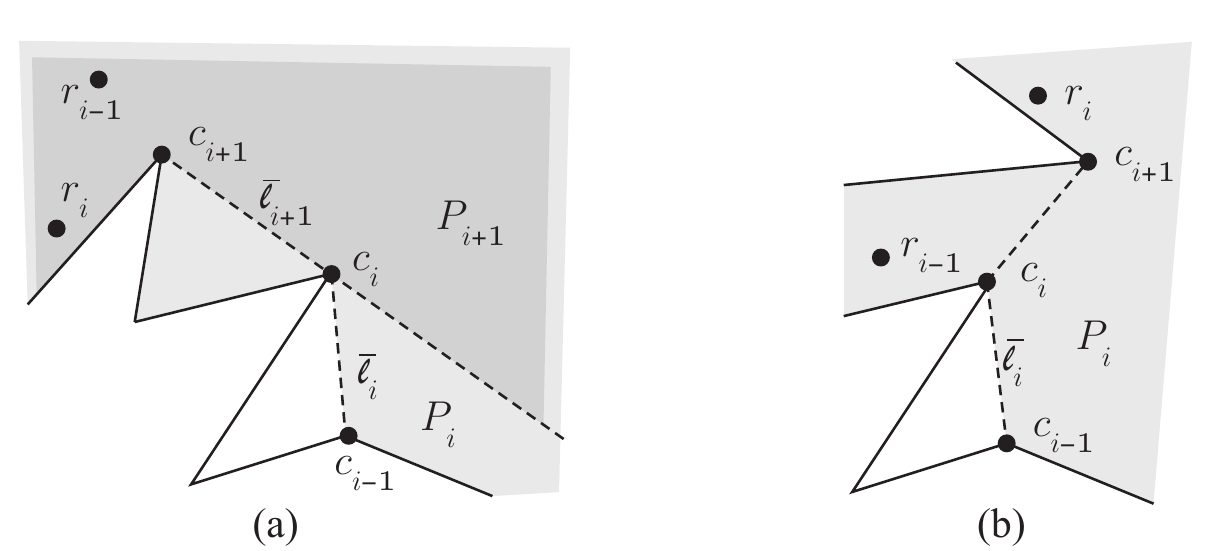}
\caption{Case 1.  (a) If $c_{i+1}$ is left of the ray $c_{i-1} c_i$ then $P_{i+1}$ (darkly shaded) is a subset of $P_i$ (lightly shaded). (b) It cannot happen that $c_{i+1}$ is to the right of the ray $c_{i-1} c_i$ because the robber could not have moved from $r_{i-1}$ to $r_i$.}
\label{fig:polygon-case1}
\end{figure}

\smallskip\noindent{\bf Case 2.} See Figure~\ref{fig:polygon-case2}. The shortest path from $c_i$ to $r_i$ makes a right turn at $c_{i+1}$. 
We consider two subcases: (a) $c_{i+1}$ is left of the ray $c_{i-1} c_i$; and (b) $c_{i+1}$ is right of (or on) the ray $c_{i-1} c_i$.
See Figure~\ref{fig:polygon-case2}. 
In case (a) $\bar \ell_{i+1}$ stops at $c_i$ and in case (b) it may happen that $\bar \ell_{i+1}$ extends past $c_i$, but in either case, segment $\bar \ell_i$ is outside $P_{i+1}$, and $P_{i+1}$ is a subset of $P_i$ and smaller by at least one vertex.
\end{proof}

\begin{figure}[ht]
\centering
\includegraphics[width=.73\textwidth]{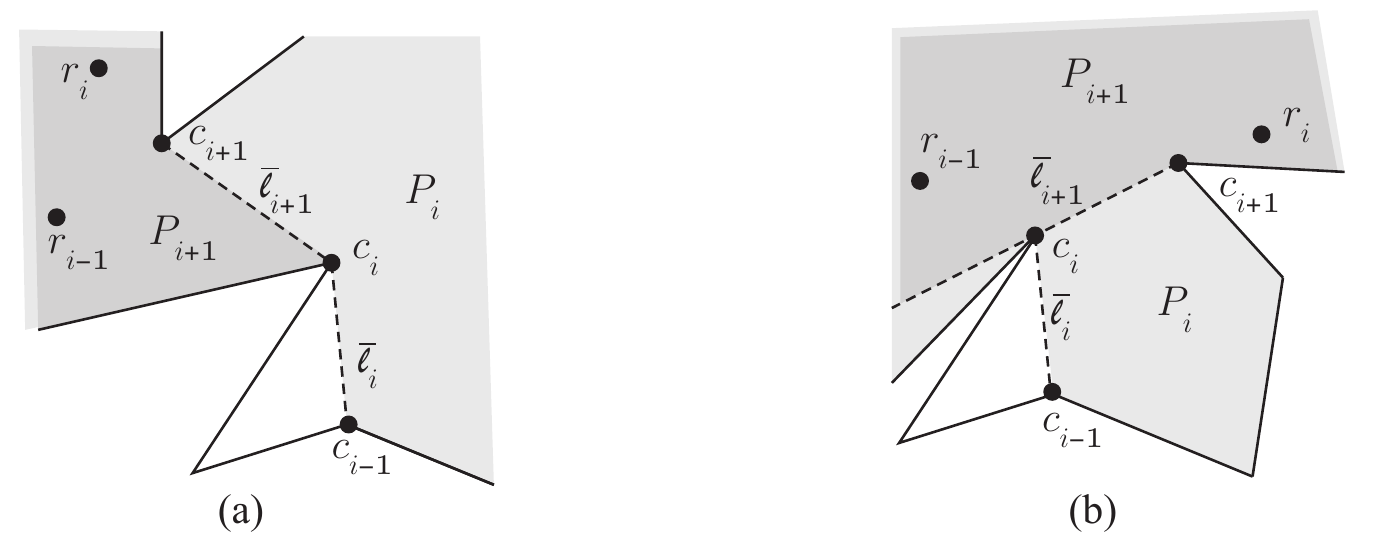}
\caption{Case 2.  (a) $c_{i+1}$ is left of the ray $c_{i-1} c_i$.  (b) $c_{i+1}$ is right of the ray $c_{i-1} c_i$. 
In either case $P_{i+1}$ (darkly shaded) is a subset of $P_i$ (lightly shaded).}
\label{fig:polygon-case2}
\end{figure}

We note %he curious similarity 
that Bhadauria et al.~\cite{Bhadauria-12} use the same cop strategy of
following a shortest path to the robber for the version of the problem
where each cop or robber move is at most distance~1.

Theorem~\ref{thm:polygon} can alternatively be proved by decomposing the polygon into $O(n^2)$ triangular regions and proving that they have an ordering with properties like a dismantlable ordering, but we do not include the proof here.
%.  See the Appendix.  \rednote{Say more about what we mean by a dismantlable ordering.}

\changed{
\subsection{Lower Bounds}
\label{sec:lower-bound}

In this subsection we discuss lower bounds on the worst case number of cop moves.  The example in Figure~\ref{fig:zig-zag} shows that the cop may need $\Omega(n)$ moves even when it may move on interior points of the polygon.  We give an example to show that this lower bound holds even when the polygon has small link diameter.

\begin{theorem} 
\label{thm:lowerbound}
There is an $n$-vertex polygon with link diameter 3 and a robber strategy that forces the cop to use $\Omega(n)$ steps.
%The number of steps cop needs to win inside a simple polygon might be in $\Omega(n)$.
\end{theorem}

\begin{figure}[ht]
\centering
\includegraphics[width=.8\textwidth]{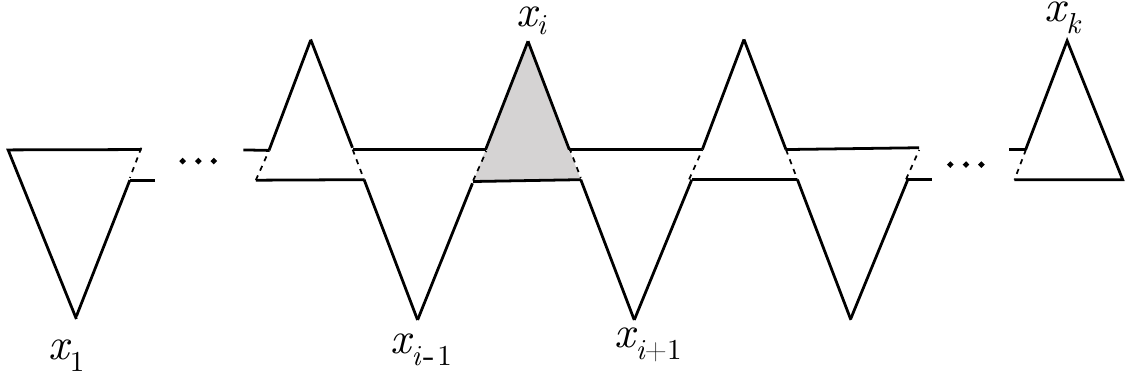}
\caption{A polygon with constant link diameter in which the robber can survive for $\Omega(n)$ steps.}
\label{fig:degen-example}
\end{figure}

\begin{figure}[ht]
\centering
\includegraphics[width=.55\textwidth]{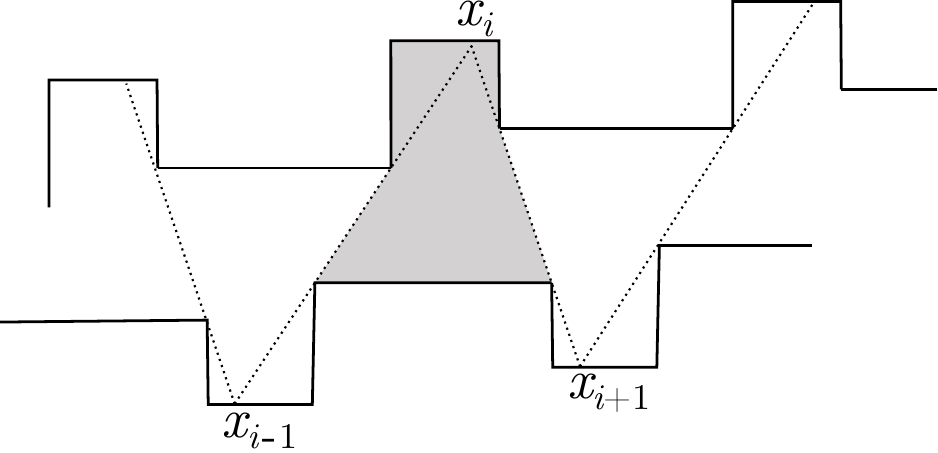}
\caption{We can adjust the edges and the dashed-line path such that the polygon becomes non-degenerate.}
\label{fig:polygon-lowerbound}
\end{figure}

\begin{proof}
We modify the zig-zag polygon from Figure~\ref{fig:zig-zag} to decrease the link diameter to 3 as shown in Figure~\ref{fig:degen-example}. The polygon consists of $k=n/3$ similar sections concatenated, where $n$ is the number of vertices of the polygon. 
%Let $k$ be the size of the dashed line shown in the figure~\ref{fig:degen-example}. Thus, $k = n/4$ that is clearly  of $O(n)$. 
We will show that the robber has a strategy to survive at least $k/2$ steps in such a polygon.

\noindent {\bf Robber Strategy:} The robber plays on points $x_i$, $1\leq i \leq k$. Initially the robber chooses the closest point $x_i$ to $x_{\lfloor{k/2}\rfloor}$ such that $x_i$ is not visible to the cop's initial position.
Observe that $x_i$ is only visible to the gray area in Figure~\ref{fig:degen-example} and the line segments $x_{i-1} x_i$ and $x_{i}x_{i+1}$, so the cop can only see at most two of the $x_i$'s from its initial position. 

The robber remains stationary until it is visible to the cop, i.e.,~when the cop enters the gray area or along the segments $x_{i-1} x_i$ and $x_{i}x_{i+1}$. Then, the robber moves to one of the neighbors $x_{i-1}$ or $x_{i+1}$, the one that is not visible to the cop. At least one of $x_{i-1}$ and $x_{i+1}$ is safe for the robber to move to in the next step as we observe that there is no point visible to all three of $x_{i-1}$, $x_i$, and $x_{i+1}$. The robber can survive at least $k/2 -1$ steps with this strategy as the game may only be terminated at either $x_1$ or $x_k$.

Note that the polygon in Figure~\ref{fig:degen-example} is degenerate, with edges that lie on the same line, but we can adjust it a little bit to resolve all degeneracies, as shown in Figure~\ref{fig:polygon-lowerbound}, at the expense of decreasing the lower bound to $n/8$. 
We need to be careful not to increase the link diameter of the polygon through the edge level changes. The changes in horizontal edge levels must be small in comparison with the width of the middle corridor. 
\end{proof}
}

%%%%%%%%%%%%%%%%%%%%%%%%%%%%%%%%%%%%%%%%%%
%%%%%%%%%%%%%%%%%%%%%%%%%%%%%%%%%%%%%%%%%%
%\clearpage
 
%\section{Cops and Robbers Inside a Curved Region}% the boundary is
%curved, not the region. Splinegons were defined, so let's use them.
\section{Cops and Robbers Inside a Splinegon}
\label{sec:region}

In this section we consider the cops and robbers game in a simply
connected region with curved boundary, specifically a \emph{splinegon}
$R$ whose boundary consists of $n$ smooth curve segments that each lie on their own convex hull.
Other natural assumptions (such as algebraic
curves or splines of limited degree, or other curves of constant
complexity) give regions that can be converted to splinegons with a
constant factor overhead by cutting at points of inflection and
points with vertical tangents. Assume that tangents in a given
direction and common tangents between curve segments can be computed.
A \emph{vertex} is an endpoint between two curve segments.
%An endpoint between two curve segments is a 
%\emph{vertex} if the tangents to the two curves form an angle other than $180^\circ$.

\remove{
Clearly we need some conditions on the curve to ensure the cop can win.
 We suggest playing in a ``splinegon''. 
Assume that the splinegon $R$ consists of $n$ curve segments that each
lie on their own convex hull.
}

We need another assumption to avoid an infinite game where the cop
gets closer and closer to the robber but never reaches it.  This
occurs, for example, when two curves meet tangentially at a vertex as
in Figure~\ref{fig:crescent-moon}---in fact, in this situation a
robber at a vertex avoids capture by remaining stationary.  
One possibility is to assume that the robber is captured when the cop gets sufficiently close (within some $\epsilon$).
Instead, we will make the assumption that the 
%We will assume that the 
link distance between any two points in the
splinegon $R$ is finite, and bounded by $d$.  

%We will assume a bounded ``minimum cutoff area''.  Suppose we have a tangent line to the boundary such that the two points where the line exits the region occur on the same (convex) curve segment (for example, the internal segments of the path in Figure~\ref{fig:crescent-moon}).  The region between the tangent line and the boundary is called a \emph{cutoff}.  We will assume that every cutoff has area at least some $a_{\rm min} >0$.

\begin{figure}[ht]
\centering
\includegraphics[width=.21\textwidth]{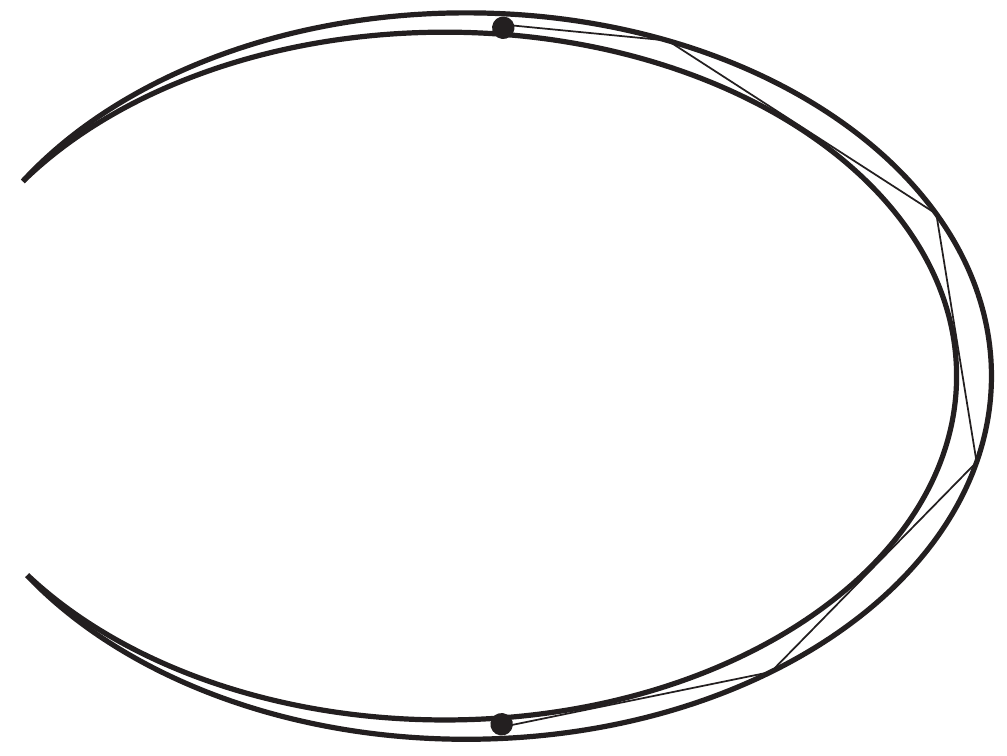}
\caption{The number of cop moves may be infinite even when $n=2$.  
Truncating the vertices makes the game finite but the 
number of moves may depend on the link diameter.}
\label{fig:crescent-moon}
\end{figure}

With these assumption---a splinegon $R$ of $n$ curved segments with link diameter $d$---we 
%with cutoff area bounded below by $a_{\rm min}$---we
prove that the cop always wins, and does so in $O(n^2 + d)$ steps.  But first, we show through additional examples that the strategy must be a little more complex than in the polygonal case. 

\subsection{The Cop Strategy}
A main difference from the polygonal case is that the cop
may need to move to interior points in order to win.
Figure~\ref{fig:curved-triangle}, for example, shows a region in
which the robber can win if the cop always stops on the boundary. 

\begin{figure}[h]
\centering
\includegraphics[width=.21\textwidth]{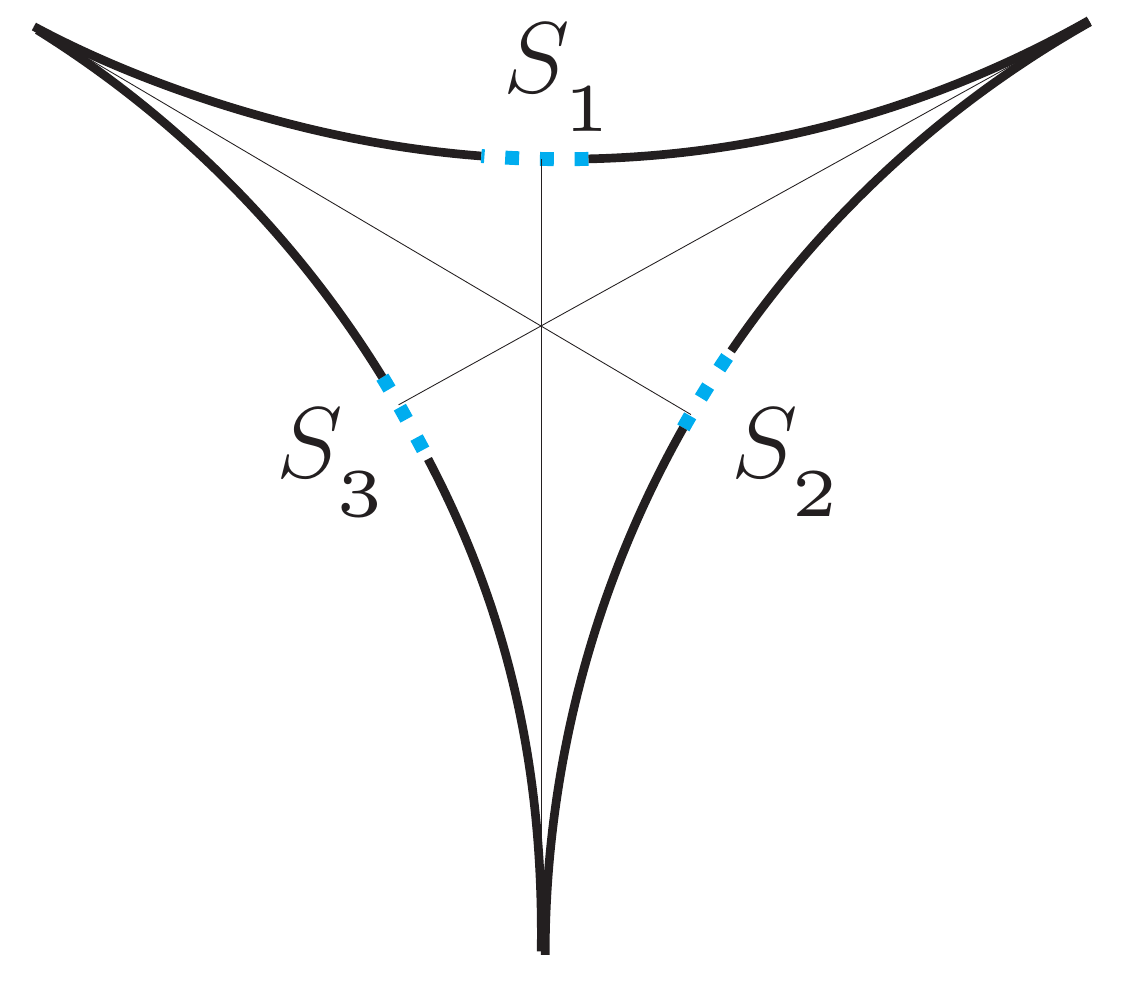}
\caption{If the cop plays only on the boundary then the robber can win: the robber's strategy is to play on the middle dashed portions of the boundary and always move to the same curve $S_i$ that the cop is on.  In our cop strategy the cop would move to the endpoint tangents (drawn as thin lines).}
\label{fig:curved-triangle}
\end{figure}

\begin{figure}[ht]
\centering
\includegraphics[width=.3\textwidth]{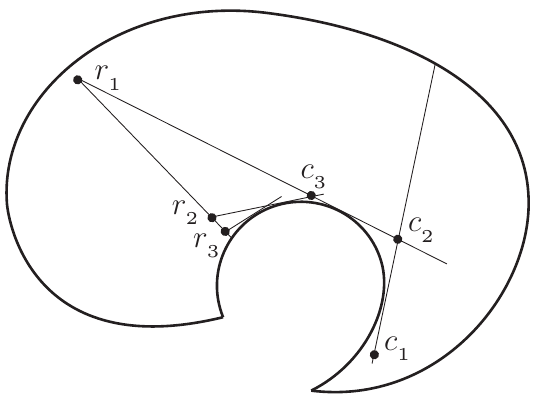}
\caption{If the cop moves only far enough to see the robber then the robber can win because it can force the cop to take smaller and smaller steps.}
\label{fig:cop-sees-robber}
\end{figure}

Our strategy is that the cop  starts off along the first straight segment of the shortest path to the robber's current position. % (or tangent to the first curve in case there is no straight segment).
However, if this segment is tangent to a concave curve of the shortest path then the cop should move 
%the cop reaches (or starts from) a \rednote{tangent of} a concave curve of the shortest path, it should move 
further, into the interior of the polygon.  How far should the cop go?
It is tempting stop the cop when it can see the robber, but
Figure~\ref{fig:cop-sees-robber} shows that this strategy fails---the
cop should move farther.
Figure~\ref{fig:godfrieds-region} shows
there is also a danger of moving the cop too far.

\begin{figure}[ht]
\centering
\includegraphics[width=.25\textwidth]{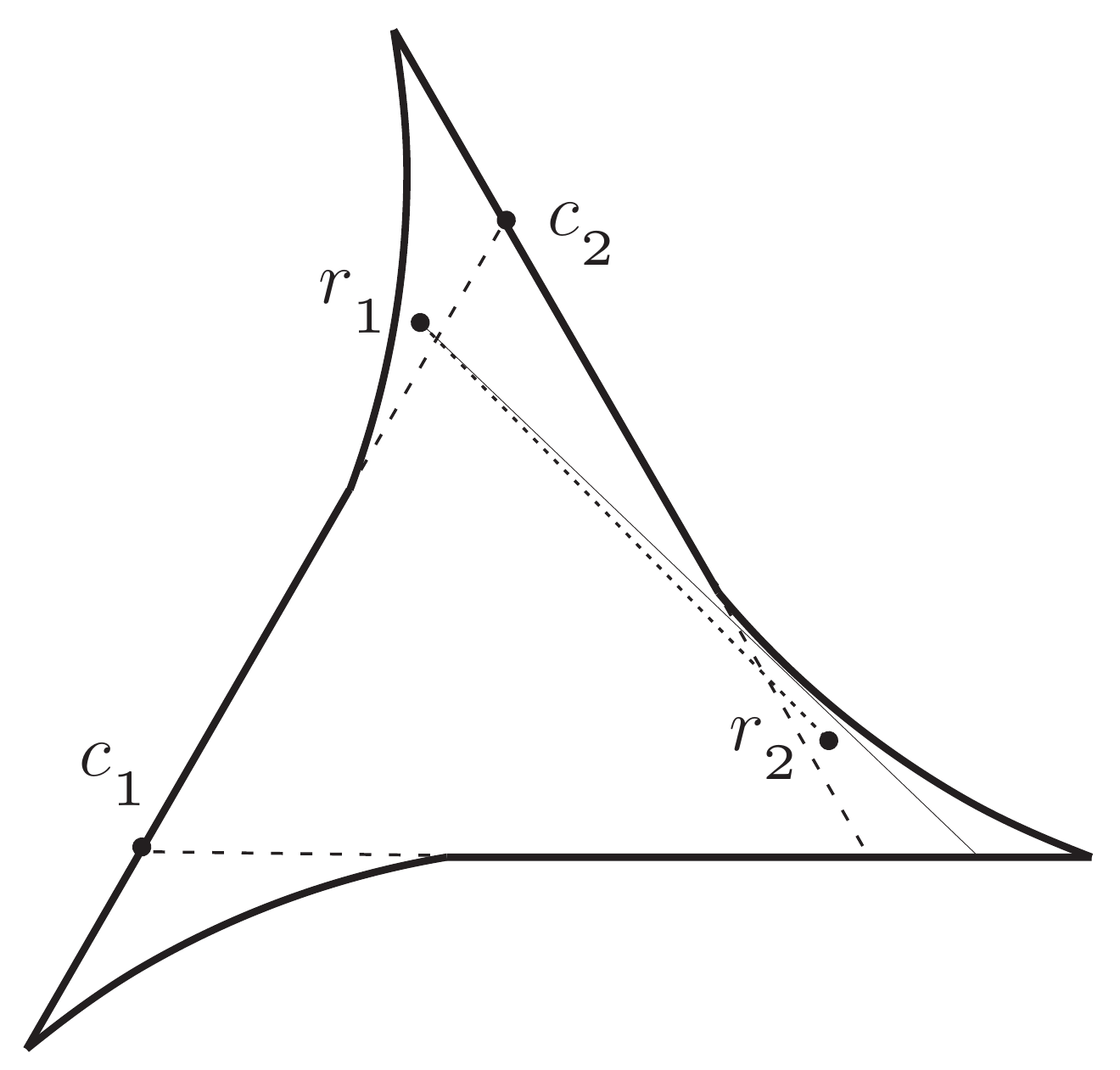}
\caption{If the cop moves too far then the robber can win: the cop
  moves from $c_1$ to $c_2$, the robber moves from $r_1$ to $r_2$
  (dotted line) and then this can be repeated around the polygon. In
  our proposed strategy the cop would not move from $c_1$ all the way
  to $c_2$---it would stop at a robber exit line (drawn as a thin line).}
\label{fig:godfrieds-region}
\end{figure}

%%%%%%%%%%%%%%%%%%%%%%%%%%
% cop strategy

In our strategy the cop will stop at 
certain lines inside the splinegon.  We first state the cop strategy in terms of these lines, and then define the lines. 
We use the notation $c_{i-1}$ for the cop's position and $r_{i-1}$ for the robber's position at the start of round $i$.  Their initial positions are $c_0$ and $r_0$.  
Recall that each round begins with a cop move.

\smallskip\noindent{\bf Cop Strategy for Round $i$.}  If the cop sees
the robber, it moves to the robber's position and wins.  Otherwise,
define the cop's next position, $c_i$ as follows:
Compute the shortest path from the cop's current position, $c_{i-1}$,
to the robber's current position, $r_{i-1}$. 
Let $\ell_i$ be the ray along the first straight segment of this  shortest path, or, if the shortest path begins with a curve, let $\ell_i$ be the tangent to this curve.
Let $b_i$ be the first point where the shortest path diverges from
$\ell_i$.  Then $b_i$ lies on the boundary of the splinegon $R$.
If  $b_i$ is not a splinegon vertex, let $\gamma_i$ be the boundary curve containing $b_i$.
If $b_i$ is a splinegon vertex then there are two boundary curves incident to $b_i$, and we let $\gamma_i$ be the one on the far side of $b_i$ with respect to the direction of $\ell_i$.
%not visible from $c_{i-1}$.  \rednote{Careful, it may be that neither curve is visible from $c_{i-1}$.}
By reflection if necessary, assume that the path starts upward and turns left,
as depicted in  Figure~\ref{fig:curvy-cop-move}.  

If $b_i\ne c_{i-1}$  and $b_i$ is a vertex,  then define $c_i$ to be $b_i$.  
(This matches the polygonal case.)
%\bluenote{When $b_i$ is a vertex maybe we would stop anyway because of the endpoint tangent of $\gamma_i$ but it seems worth saying it explicitly.}
Otherwise $\ell_i$ is tangent to $\gamma_i$, so define $c_i$ to be the first point on the ray $\ell_i$, past $b_i$, where $\ell_i$ intersects a 
\emph{common tangent} or a \emph{robber exit line} or
%\emph{tight line} or 
touches the splinegon boundary.  
%[This is to avoid the cop doubling back on the same line.  Perhaps it is covered by the cop walking along an inflection line and going past the endpoint.]

\begin{figure}[ht]
\centering
\includegraphics[width=.25\textwidth]{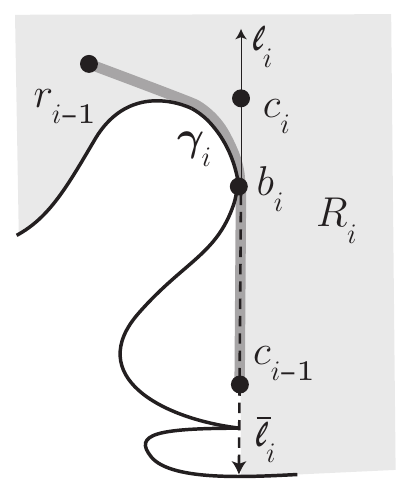}
\caption{The cop move, showing the shortest path from $c_{i-1}$ to
  $r_{i-1}$ (thick grey path), the first straight segment of this path
  $c_{i-1}b_i$ upward along ray $\ell_i$, the new cop position $c_i$,
  the downward segment $\bar\ell_i$ (dashed) and the active region $R_i$ (lightly shaded).
}
\label{fig:curvy-cop-move}
\end{figure}

%%%%%%%%%%%%%%%%%%%%%%%%%%%%%%%%%%%%%%%%

We now define \emph{common tangents} and \emph{robber exit lines}. Refer to Figure~\ref{fig:defns}.
A \emph{common tangent} is a line segment that is tangent to $R$ at two points and extends in both directions until it exits the region.
%The part of an inflection line between the two tangents can appear on shortest
%paths and the line is important for limiting visibility. 
At each endpoint of each curve we have an \emph{endpoint tangent}---the tangent to the curve through the endpoint.  An endpoint tangent extends in both directions until it exits the region.
We count endpoint tangents as common tangents.
There are $O(n^2)$ common tangents, because a curve has at most two common tangents with any other curve or vertex. 
%(including its own endpoints---we count endpoint tangents as common tangents, too).

We define \emph{robber exit lines} 
relative to the current robber and cop positions, using the notation from the cop strategy above. 
See Figure~\ref{fig:defns}. 
Consider segments that start at $r_{i-1}$ and are tangent to $R$,
ending at the tangent point.    Among these,  a \emph{robber exit
  line} is one that crosses ray $\ell_i$ such that the tangent point is on the far side of the segment with respect to the direction of $\ell_i$.  
If we extend a robber exit line past its tangent point to the region boundary we obtain a \emph{bay} of points not visible from the robber position.  Note that every bay contains a vertex of the region---either the tangent point itself is a vertex or the tangent point is on a reflex curve, and we must change curves before the end of the bay.
 
\begin{figure}[ht]
\centering
\includegraphics[width=.5\textwidth]{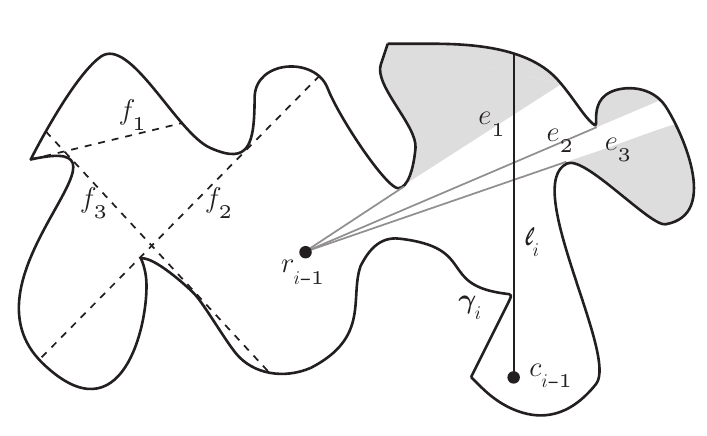}
\caption{Lines $f_1, f_2$ and $f_3$ are three of the many common tangents.  Lines $e_1, e_2$ and $e_3$ go through $r_{i-1}$ and are tangent to $R$; of those, $e_2$ is the only exit line. Lightly shaded regions are the bays.}
\label{fig:defns}
\end{figure}

\remove{
We define \emph{tight lines} to be common tangents plus the robber exit lines of the current robber's position. 
%the set of all inflection lines inside $R$ and the exit lines of current position of the robber.
With this definition, our description of the cop strategy is complete.
}

With these definitions of common tangents and robber exit lines, we have completely specified the cop strategy.
%This completes the description of the cop strategy.
We note that the cop's move can be computed in polynomial time assuming we have constant time subroutines to compute common tangents and tangents at a given point.  We can preprocess to find all common tangents.  For a given robber position, we can find all robber exit lines in polynomial time.  We can find shortest paths in the splinegon $R$ in linear time using the algorithm of~\cite{melissaratos1992}.  With this information, we can find the next cop position.  
A straightforward implementation takes $O(n^2)$ time, though this can probably be improved.
 
%A \emph{tight line} is a line segment inside the region that goes through two reflex vertices, or goes through one reflex vertex and is tangent to a concave piece of the boundary, or is tangent to two concave pieces of the boundary.
%Our assumption about the region implies that there are a finite number, $\tau$, of key lines. 
%The number of cop moves will depend on $\tau$.
%We will also stop the cop at \emph{robber tight lines}.  A \emph{robber tight line} is a line segment inside the region that goes through the robber's current position and is either tangent to a concave piece of the boundary or goes through a reflex vertex.

%%%%%%%%%%%%%%%%%%%%%%%%%%%%%%%%%%%%%%%%

\remove{%%%%%% 
\begin{theorem}\label{thm:splinegon}
The cop wins the cop and robber game on the points inside a splinegon $R$ by the following strategy at each step $i$, $i\geq 1$: Compute the shortest path from cop's current position $c_{i-1}$ to robber's current position $r_{i-1}$. Move on the first edge of the shortest path and continue until it hits the first \emph{key line} or exits $R$.  
\end{theorem}
} %%%%%%%%%%%%%

%%%%%%%%%%%%%%%%%%%%%%%%%%%%%%%%%%%%%%%%%%%%%
\subsection{The Cop Wins}
In order to prove %theorem~\ref{thm:splinegon} 
that the cop wins using the strategy specified in the previous section,
we  first show that each cop move restricts the robber to a smaller subregion.
%, using a definition of line segment $\ell_i$ and subregions similar to the polygon case (see Figure~\ref{fig:curvy-cop-move}). % However, the difficulty is in proving that this process terminates in a finite number of moves. 
Then we  show that the number of steps the cop needs to win is 
$O(n^2 + d)$ where $n$ is the number of segments and $d$ is the link diameter of the region. 

We begin by defining the subregion that the robber is restricted to during and after round $i$.
%\rednote{Define via exclusion bay.}
Define $\bar\ell_i$ to be the directed segment that starts at  
$b_i$, goes opposite ray $\ell_i$  through $c_{i-1}$, and stops at the
first boundary point for which every $\epsilon>0$ neighborhood
contains a boundary point on the opposite side of $\bar\ell_i$ from
$\gamma_i$ at $b_i$ (i.e.,~where part of the boundary is to the
left of the downward directed segment $\bar \ell_i$ in Figure~\ref{fig:curvy-cop-move}.)
%where part of the boundary goes to the left.
The segment $\bar\ell_i$ starts and ends on the boundary so it
cuts the region into two (or more) pieces; define the \emph{active region}, $R_i$, to be the piece that contains $r_{i-1}$.  Define the \emph{exclusion region} to be its complement in $R$.
Observe that the robber cannot exit $R_i$ in round $i$, i.e., $r_i$ is inside $R_i$.  This is because
$r_{i-1}$ is on the wrong side of the line through $\bar \ell_i$.

 We  prove below in Lemma~\ref{lem:growing} that $R_{i+1} \subsetneq R_i$, i.e., the active region shrinks.  
 Following that, we  show that the cop wins in a finite number of steps.  The proofs are similar to the analogous results for polygons, and involve handling four cases for the left/right configuration of the cop and the robber.  
Suppose that the shortest path from $c_{i-1}$ to $r_{i-1}$ makes a left turn at $b_{i}$.
(The other case is completely symmetric.)
%Let $\gamma_i$ be the piece of the boundary curve that the shortest path follows just after $b_i$.
%Observe that the next robber position $r_i$ must be inside $R_i$, i.e.,~the robber cannot move from $r_{i-1}$ to cross $\ell_i$.
We distinguish the following cases: 

\smallskip
\noindent{\bf Case 1.} The shortest path from $c_i$ to $r_i$ makes a left turn at $b_{i+1}$.

{\bf (a)} $c_{i+1}$ is left of the ray $c_{i-1} c_i$---more precisely, in moving  from $c_{i-1}$ to $c_i$ to $c_{i+1}$ the cop turns left by an angle in the range $(0,180^\circ)$.  (Turning by $0^\circ$ will be handled in case (b).)

{\bf (b)}  $c_{i+1}$ is right of the ray $c_{i-1} c_i$---more precisely, the cop turns right at $c_i$ by an angle in $[0,180^\circ)$.

\noindent{\bf Case 2.} The shortest path from $c_i$ to $r_i$ makes a right turn at $b_{i+1}$.

{\bf (a)} $c_{i+1}$ is left of the ray $c_{i-1} c_i$---more precisely, the cop turns left at $c_i$ by an angle in $(0,180^\circ)$.

{\bf (b)}  $c_{i+1}$ is right of the ray $c_{i-1} c_i$---more precisely, the cop turns right at $c_i$ by an angle in $[0,180^\circ)$.

\smallskip

Note that the cop never turns by an angle of $180^\circ$ (doubling back) because then $b_{i+1}$ would be on the line segment between $c_i$ and $b_i$ or further along, on the ray $\bar \ell_i$.  In the first case, $b_{i+1}$ would provide a stopping point for $c_i$ according to the rule that the cop stops on the  boundary.  The second case is impossible because the robber can never move from $r_{i-1}$ to a position that would cause the cop to move onto $\bar \ell_i$.

\remove{
\rednote{by H.V.} Note that we ignored the case where $c_{i-1}$, $c_i$, and $c_{i+1}$ are collinear in both cases. In case 1 collinear situation never happens because otherwise the robber would need to exit a bay (where the path bends in turn $i$) and enter the other bay (where the path bends in turn $i+1$) that's not possible with collinear moves. In case 2 we may count collinear case as the extreme case of one of two cases (a) or (b).
}

We  begin by showing that Case 2(a) can happen only in special circumstances and 
that  Case 1(b) cannot happen at all.

\remove{%%%%%%%%%%% this is replaced by observation above
\begin{lemma}
After cop's move at step $i$, robber is limited to the subregion that is cut by $\ell_i$, called robber's region $R_i$.
\end{lemma}
\begin{proof}
According to the definition of $\ell_i$ and $R_i$, the only way to exit $R_i$ is to cross $\ell_i$.
The robber at $r_{i-1}$ may not cross $\ell_i$ in step $i$ because the shortest path from $c_{i-1}$ to $r_{i-1}$ bends left at $b_i$ and $\ell_i$ is fully behind $b_i$.
\end{proof}

%%%%%% and this is moved later on
\begin{lemma}\label{lem:growing}
The robber's region $R_i$ shrinks after each move of the cop.
\end{lemma}

\bluenote{Similar to the polygonal case, we should handle four cases that may happen after robber's move to $r_i$. In the following two lemmas we show that two cases are impossible following  cop's strategy in curved}
} %%%%%%%%%%%%%%% end remove

\begin{lemma}\label{lem:infl}
In Case 2(a) the segment $c_ib_{i+1}$ is tangent to the boundary on its left side (as well as tangent to the boundary on its right side at $b_{i+1}$).  
%Suppose that the shortest path from cop to robber begins with an
%upward segment to a left bend, and the cop moves upward to $c_i$.  If
%the robber remains in the bay defined by the vertical lid from $c_{i-1}$,  and the $c_i$ to $r_i$ shortest path in $R$
%first bends right, then the cop crossed an inflection line.
\end{lemma}
\begin{proof}
See Figure~\ref{fig:infl}.  
The segment $c_i b_{i+1}$ is tangent to the boundary curve $\gamma_{i+1}$ on its right side at point $b_{i+1}$. 
Suppose that segment $c_ib_{i+1}$ is not tangent to the boundary on its left side.  
We  show that the cop has passed a common tangent, which is a contradiction.
Move $c_i$ back towards $c_{i-1}$ while maintaining tangency with the curve $\gamma_{i+1}$.  We can move some positive amount. 
Either we reach the tangent at an endpoint of $\gamma_{i+1}$ 
% note the following subtlety -- that in this case \gamma changes to a new curve
or the segment $c_ib_{i+1}$ hits a boundary point on its left side (possibly because $c_i$ reaches $b_i$).  In either case, we  have arrived at a common tangent, so $c_i$ should have been placed here rather than further along. 
\end{proof}

\begin{figure}[ht]
\centering
\includegraphics[width=.3\textwidth]{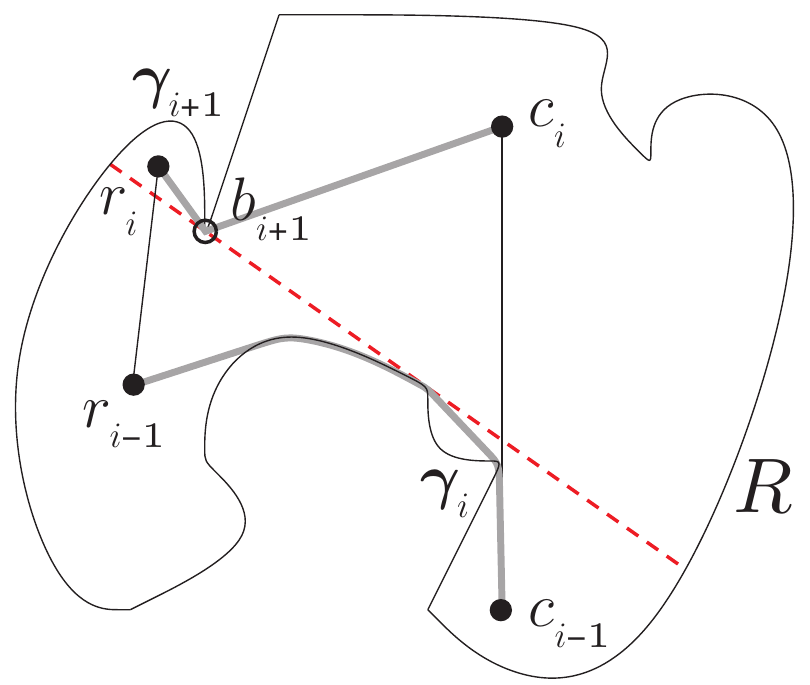}
\caption{In case 2(a) if the segment $c_i b_{i+1}$ is not tangent to the boundary on its left side then there is an earlier choice for $c_i$ (on the dashed red common tangent).}
\label{fig:infl}
\end{figure}

\begin{lemma}
\label{lem:exit}
Case 1(b) cannot occur, i.e., it cannot happen that the shortest path from $c_i$ to $r_i$ makes a left turn at $b_{i+1}$ and $c_{i+1}$ is right of the ray $c_{i-1} c_i$ by an angle in $[0,180^\circ)$.
\end{lemma}

\begin{proof}
Suppose the situation does occur.  See Figure~\ref{fig:exit}.  We  show that the cop has passed a common tangent or a robber exit line, which gives a contradiction.
Because $c_{i+1}$ is to the right of the ray $c_{i-1} c_i$, therefore the robber's move $r_{i-1}r_i$ must have crossed the line  through $c_{i-1} c_i$, say at point $x$.  
We claim that segments $r_{i-1}r_i$ and $c_{i-1} c_i$ intersect. 
First note that $x$ lies after $b_i$ along the ray $c_{i-1} c_i$.  We must show that $c_i$ lies after $x$ along this ray. 
If $c_i$ lies before $x$, then there is a two-link path inside the region, $c_i, x, r_i$ that turns right at $x$.  Shortening this to a locally shortest path, we obtain the shortest path from $c_i$ to $r_i$ that makes a first turn to its right, contradicting our assumption. 

\begin{figure}[ht]
\centering
{
\subfigure[]{\includegraphics[width=.3\textwidth]{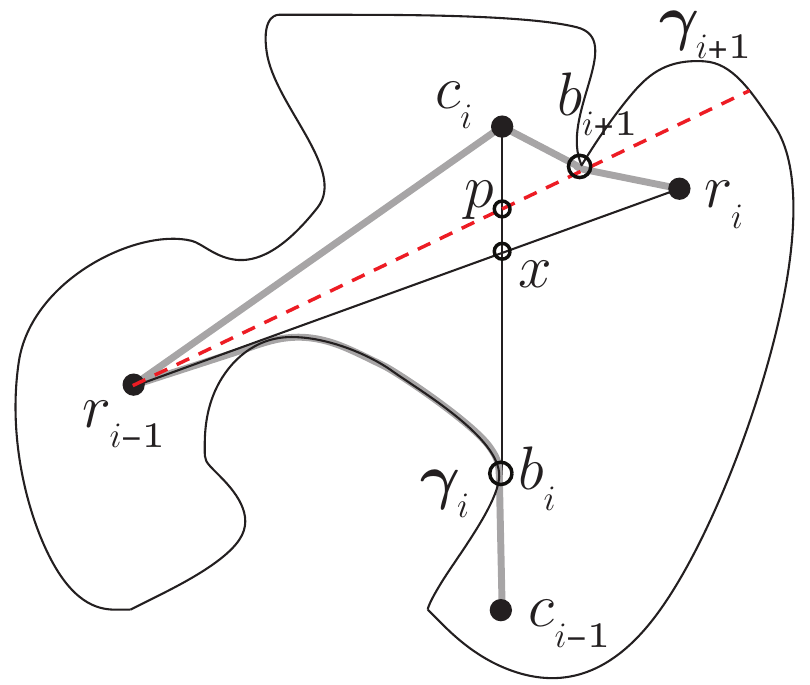}\label{fig:exit}}
\hspace{.5in}
\subfigure[]{\includegraphics[width=.3\textwidth]{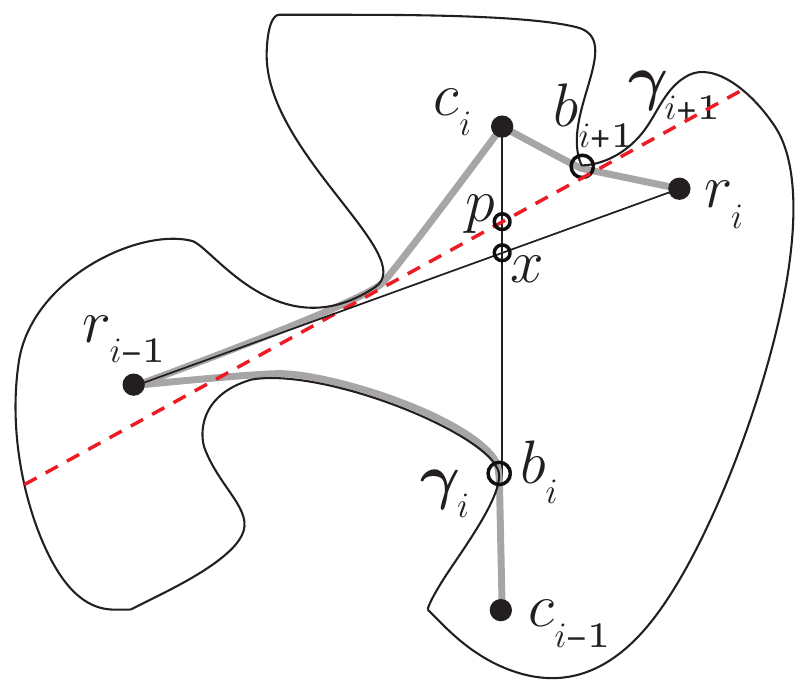}\label{fig:case1b-robber}}
}
\caption{Case 1(b) cannot occur in these situations: moving $p$ from $c_i$ towards $b_i$ while maintaining tangency of $\tau$ (dashed red) with $\gamma_{i+1}$ we encounter: 
\subref{fig:exit} a robber exit line or
\subref{fig:case1b-robber} a common tangent.
}
\label{fig:case1b}
\end{figure}

Define $\sigma$ to be the shortest path from $r_{i-1}$ to $c_i$. 
The segment $c_i b_{i+1}$ is tangent to the curve $\gamma_{i+1}$.  We will now
move point $p$ from $c_i$ towards $b_i$, maintaining a segment $\tau$ through $p$ tangent to the curve $\gamma_{i+1}$.  In the other direction, $\tau$ extends to its intersection point with $\sigma$.
See Figures~\ref{fig:case1b} and \ref{fig:case1b-middle-infl}, where $\tau$ is drawn as a dashed red line.
If $\tau$ reaches an endpoint tangent of $\gamma_{i+1}$ then we have a common tangent and the cop should have stopped at point $p$.
Otherwise, the segment $\tau$ must at some point lose contact with $\sigma$ and we claim that this can happen only because of one of the following:
 \begin{itemize}
\item The segment $\tau$ intersects $\sigma$ at $r_{i-1}$; this is a robber exit line.  See Figure~\ref{fig:exit}.
\item The segment $\tau$ becomes tangent to $\sigma$; this is a common tangent.  See Figure~\ref{fig:case1b-robber}.
\item The segment bumps into the region boundary (possibly at point $b_i$); this is a common tangent. 
See Figure~\ref{fig:case1b-middle-infl}. 
\end{itemize}
In all cases, the cop should have stopped at point $p$ because of the common tangent or robber exit line $\tau$. 
\end{proof}
%
%Otherwise $p$ reaches $b_i$, and this is a common tangent, so we would have placed $c_i$ at $b_i$.  See Figure~\ref{fig:case1b-b}. 
\remove{
The segment $c_i b_{i+1}$ is tangent to the curve $\gamma_{i+1}$.  We will now
move point $p$ from $c_i$ towards $b_i$, maintaining tangency with the curve $\gamma_{i+1}$.
We extend the segment past $p$ to the point where it reaches the segment $r_{i-1}c_i$.
If any of the following events occur, we would have placed $c_i$ at $p$:
\begin{itemize}
\item The segment reaches an endpoint tangent of $\gamma_{i+1}$: this is an inflection line.
\item The extension of the segment passes through $r_{i-1}$: this is a robber exit line.  See Figure~\ref{fig:case1b-robber}.
\item The segment bumps into the boundary of the region: this is an inflection line.  See Figure~\ref{fig:case1b-bump}.
\end{itemize}
Otherwise $p$ reaches $b_i$, and this is an inflection line, so we would have placed $c_i$ at $b_i$.  See Figure~\ref{fig:case1b-b}.
} 

\begin{figure}[ht]
\centering
{
\subfigure[]{\includegraphics[width=.3\textwidth]{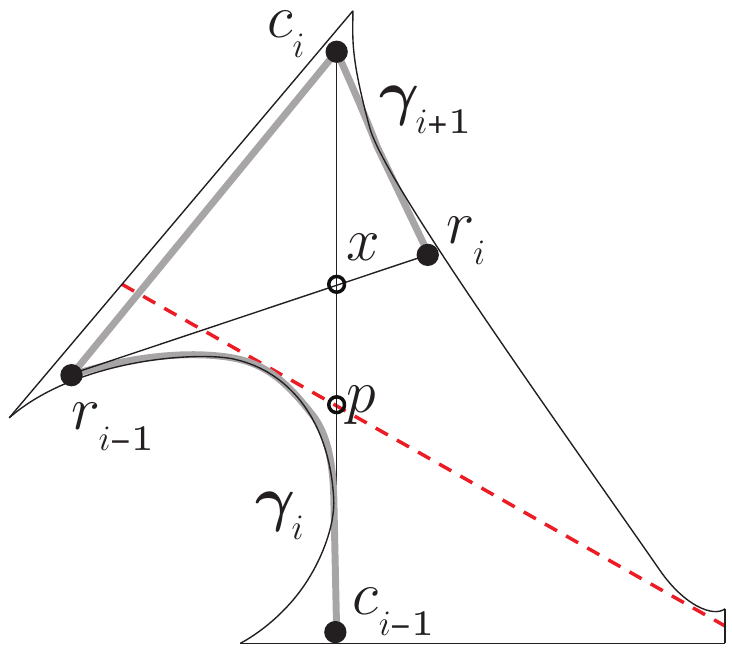}\label{fig:case1b-bump}}
\hspace{.5in}
\subfigure[]{\includegraphics[width=.3\textwidth]{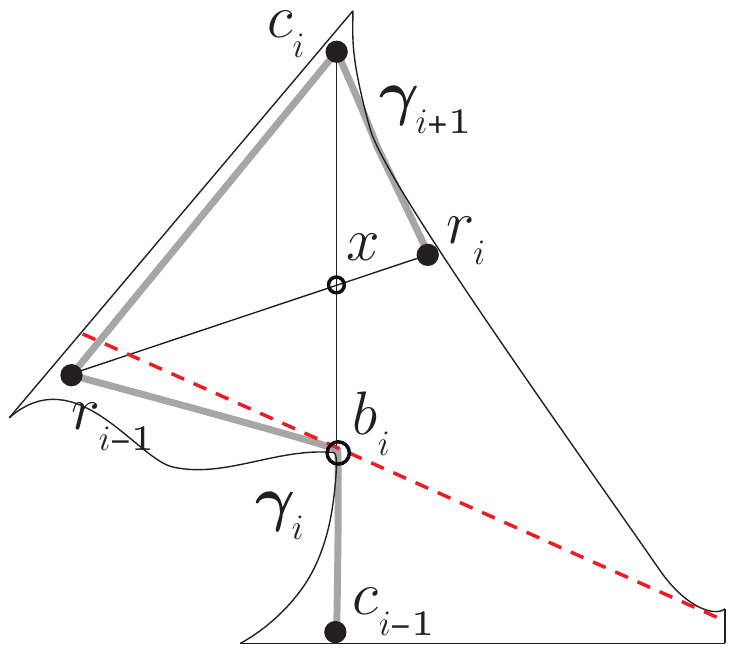}\label{fig:case1b-b}}}
\caption{Case 1(b) cannot occur in these situations: moving $p$ from $c_i$ towards $b_i$ while maintaining tangency of $\tau$ (dashed red) with $\gamma_{i+1}$ we encounter a common tangent by bumping into the region boundary %(a) before $p$ reaches $b_i$, or (b) at $b_i$.  
\subref{fig:case1b-bump} before $p$ reaches $b_i$, or
\subref{fig:case1b-b} at $b_i$.
}
\label{fig:case1b-middle-infl}
\end{figure}

%Hamide's text:
%When $p$ arrives at $x$ none of the following events have happened:
%\begin{itemize}
%\item The tangent to $\gamma_{i+1}$ reaches a vertex: this would be an inflection line
%\item The extension of the tangent passes through $r_{i-1}$: this would be a robber exit line
%\end{itemize}
%Otherwise, $p$ reaches $x$, and the extension of the segment still crosses $r_{i-1}c_i$.  See Figure~\ref{fig:case1b-middle-infl}. 
%We continue to move $p$ from $x$ towards $b_{i-1}$.  
%
%The tangent line looks like figure ? (it passes through quarters 2 and 4 in crossing at $x$)
%Continue pulling down $p$ from $x$ to $b_i$, while the extension of the tangent crosses segment $r_{i-1}x$ before going outside the splinegon. We will stop before or on $b_i$ because $b_i$ is on the boundary and for all points after $b_i$ the extension of the tangent does not cross $r_{i-1}x$. The last point on $xb_i$ that has this property should be one of the following:
%\begin{itemize}
%\item The tangent reaches $r_{i-1}$: It forms an exit line
%\item The extension of the tangent does not cross $r_{i-1}x$ before passing $r_{i-1}$ due to a bump on the left side of the tangent (or its extension): a reflection line 
%\end{itemize}
%In both situations the cop should have stopped at $p$ according to the strategy.

\remove{ % Jack's version (Lemma 2).}     
Assume that the robber exits the bay, which means that it changes sides with
respect to the cop's vertical motion, and that the $c$ to $r$ shortest path
makes a bend. If this bend is to the left, then the cop's motion also
changed side of the e robber of the robber's, since the cop enters $R$ from
below, and ends on the $c$ to $r$ path which lies above the robber's motion,
since it turns left. Therefore, assume the motion segments of $c$ and $r$
intersect. We must show that the cop crossed an inflection or exit edge.

The visibility polygon of $c$ defines a bay containing~$r$. Either that bay
contains a lid of~$R$ whose exit line intersects the motion of~$c$ above the
motion of~$r$, or a portion of a curve $\alpha$ that is invisible from $c$
at the bend was visible from $r$. As in the previous proof, moving $c$ back
and maintaining tangency with $\alpha$ must reach an inflection edge---in
this case, since there is no exit line, by having the point of tangency
reach the endpoint of $\alpha$ before $c$ reaches the motion segment of $r$.
}

We are now ready to show that the active region shrinks.  

\begin{lemma}\label{lem:growing}
%The robber's region $R_i$ shrinks after each move of the cop.
The active regions satisfy $R_{i + 1} \subsetneq R_i$.
\end{lemma}
\begin{proof}%[Proof of Lemma~\ref{lem:growing}]
Assume that the shortest path from $c_{i-1}$ to $r_{i-1}$ makes a left turn at $b_{i}$, and consider the  cases as listed above.  

\noindent{\bf Case 1(a).}  The shortest path from $c_i$ to $r_i$ makes a left turn at $b_{i+1}$ and $c_{i+1}$ is left of the ray $c_{i-1} c_i$.  See Figure~\ref{fig:1a}. 
The ray $\bar \ell_{i+1}$ from $b_{i+1}$ through $c_i$ intersects $\ell_i$ at $c_i$, and is therefore completely contained in the active region $R_{i}$.  Furthermore, the open segment $c_{i-1}c_i$ is outside $R_{i+1}$ but inside $R_i$.
Thus $R_{i+1} \subsetneq R_i$.
%The back extension of ray $c_{i+1}b_{i+1}$, i.e. $\ell_{i+1}$ intersects $\l_i$ at $c_i$. It means $\bar \ell_i$ is fully excluded in $R_{i+1}$ 
%Then $R_{i+1} \subsetneq R_i$.

\begin{figure}[ht]
\centering
\includegraphics[width=.4\textwidth]{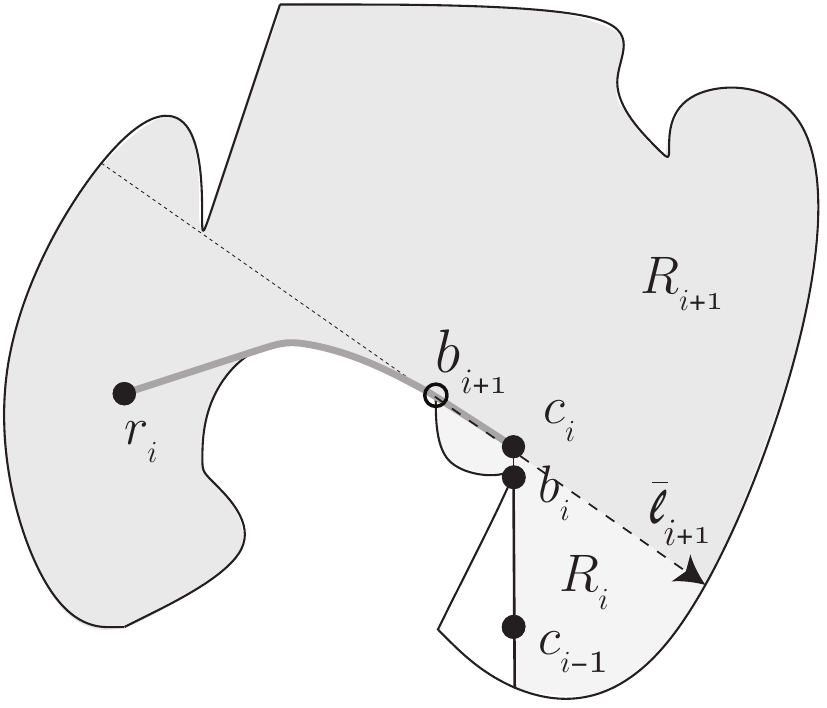}
\caption{Case 1(a).}
\label{fig:1a}
\end{figure}

\noindent{\bf Case 1(b).}  The shortest path from $c_i$ to $r_i$ makes a left turn at $b_{i+1}$
 and $c_{i+1}$ is right of the ray $c_{i-1} c_i$.
This case cannot occur by Lemma~\ref{lem:exit}.

\noindent{\bf Case 2(a).}  The shortest path from $c_i$ to $r_i$ makes a right turn at $b_{i+1}$ and $c_{i+1}$ is left of the ray $c_{i-1} c_i$.  See Figure~\ref{fig:2a}.
By Lemma~\ref{lem:infl}, 
the segment $c_ib_{i+1}$ is tangent to the boundary on its left side, say at point $p$.
The ray $\bar \ell_ {i+1}$ that defines the active region extends from $b_{i+1}$ to $p$.  It's extension 
goes through $c_i$, so it
is contained in $R_i$.
Furthermore, the open segment $c_{i-1}c_i$ is outside $R_{i+1}$ but inside $R_i$. 
Therefore $R_{i+1} \subsetneq R_i$.
%Let's the segment $\bar \ell'_{i+1}$ to be the directed segment from $p$ toward $c_i$. $\bar \ell'_{i+1} \subset \bar \ell_{i+1}$ and  movement from $c_i$ to $p$ is similar to the situation in case 1(a). Thus, with the same argument as case 1(a) $R'_{i+1} \subsetneq R_i$, where $R'_{i+1}$ is the active region corresponding to $\bar \ell'_{i+1}$. On the other hand, $R_{i+1} \subseteq R'_{i+1}$ since $\bar \ell'_{i+1} \subseteq \bar \ell_{i+1}$. Thus $R_{i+1} \subsetneq R_i$.

\begin{figure}[ht]
\centering
{
\subfigure[]{\includegraphics[width=.35\textwidth]{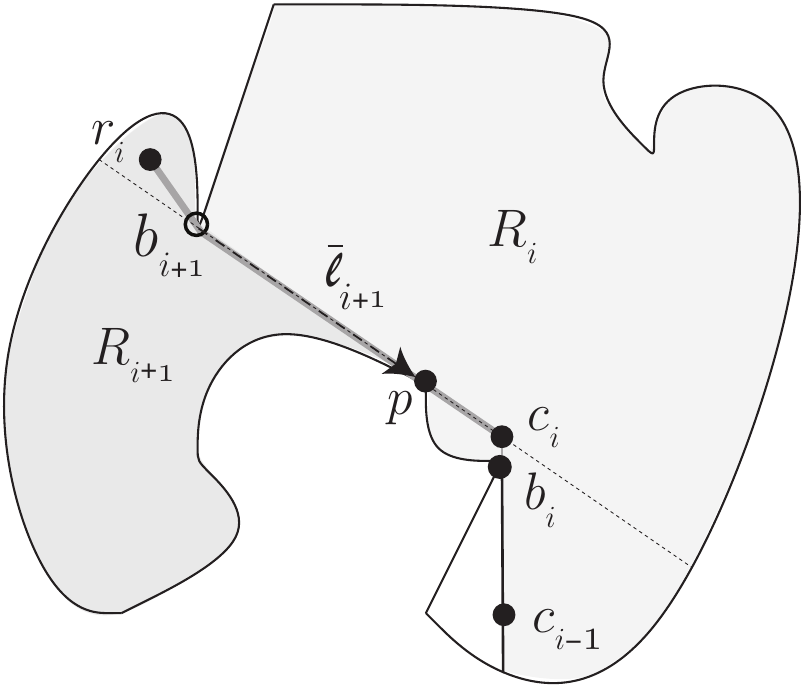}\label{fig:2a}}
\hspace{.5in}
\subfigure[]{\includegraphics[width=.35\textwidth]{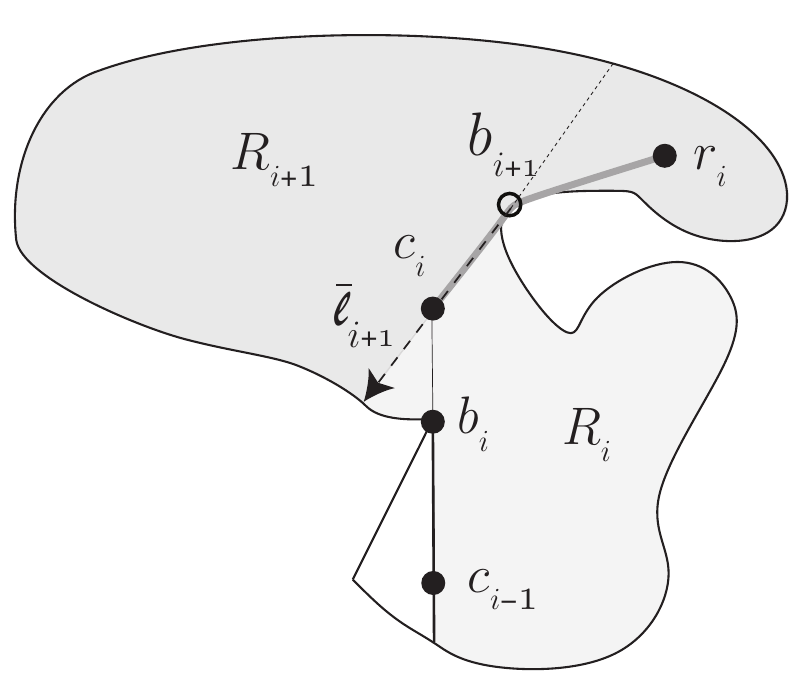}\label{fig:2b}}}
\caption{
\subref{fig:2a} Case 2(a).
\subref{fig:2b} Case 2(b).}
\end{figure}

\noindent{\bf Case 2(b).}  The shortest path from $c_i$ to $r_i$ makes a right turn at $b_{i+1}$ and $c_{i+1}$ is right of the ray $c_{i-1} c_i$ (or on the ray).  See Figure~\ref{fig:2b}. 
The ray $\bar \ell_ {i+1}$ intersects $\ell_i$ at $c_i$, and is contained in $R_{i}$.  
Furthermore, the open segment $c_{i-1}c_i$ is outside $R_{i+1}$ but inside $R_i$.  Thus $R_{i+1} \subsetneq R_i$.
\end{proof}

%\begin{figure}[ht]
%\centering
%\includegraphics[width=.3\textwidth]{Figures/Case2b-new.pdf}
%\caption{Case 2(b).}
%\label{fig:2b}
%\end{figure}

\remove{%%%%%%%%%% Hamide's proof based on Jack's
Depending on the robber's move at step $i$ from $r_{i-1}$ to $r_i$ and the shortest path from $c_i$ to $r_i$ in step $i+1$ we have following cases:

Case 1: 
Robber remains in the bay defined by the vertical lid from $c_{i-1}$, and the $c_i$ to $r_i$ shortest path 
first bends right. This case may not happen according to Lemma~\ref{lem:infl}.

Case 2:
Robber remains in the bay defined by the vertical lid from $c_{i-1}$, and the $c_i$ to $r_i$ shortest path first bends left. 
In this case, $\ell_i$ is fully contained in region excluded by $\ell_{i+1}$ and region $R_i$ has shrunken. See Figure~\ref{fig:shrink}\rednote{TODO:add fig}.

Case 3: 
Robber moves outside the bay defined by the vertical lid from $c_{i-1}$, and the $c_i$ to $r_i$ shortest path first bends left.
This case is impossible according to lemma~\ref{lem:exit}.

Case 4:
Robber moves outside the bay defined by the vertical lid from $c_{i-1}$, and the $c_i$ to $r_i$ shortest path first bend right.
It this case, $\ell_i$ is fully contained in region excluded by $\ell_{i+1}$, so the region that the robber may move has shrunken. See Figure~\ref{fig:shrink}.
} %%%%%%%%%%%%%% end remove

\remove{ %%%%%%% Make this an observation inside the main proof
\begin{lemma}\label{lem:not-reusable}
The cop may not stop on a inflection line that enters exclusive region.
\end{lemma}

\begin{proof}
Assume by contradiction that in step $i$ the cop moves from $c_{i-1}$ to $c_i$, where an inflection line $f$ intersects $c_{i-1}c_i$, and $f$ also enters exclusive region. Segment $\ell_i$ shares the same line as $c_{i-1}c_i$ and $c_i\notin \ell_i$ by definition. If $f$ comes from the exclusive region it has to intersects $\ell_i$ because $\ell_{i-1}$ is fully behind $\ell_i$ by lemma~\ref{lem:growing}. It means that the line of $\ell_i$ should intersects inflection line $f$ in two different points that is impossible. 
\end{proof}
}

%%%%%%%%%%%%%%%%%%%%%%%%%%%%%%%%%%%%%%%
Finally we prove our main result.

\begin{theorem}
Inside a splinegon of $n$ curve segments with link diameter $d$, the
cop wins the cops and robbers game in 
$O(n^2 + d)$ moves. 
\end{theorem}
\begin{proof}
We  argue that at each step between the first and the last, the active region shrinks in some discrete way.  Define the \emph{newly excluded region}, $E_i$, to be $R_{i} - R_{i+1}$. 
In order to take care of collinearities, we will include the boundary of $R_{i+1}$ but exclude the boundary of $R_i$.  %, including the boundary of $R_{i+1}$ but excluding the boundary of $R_i$.
If the two boundaries intersect, the intersection point is not included in $E_i$.
%then they intersect at point $c_i$, \rednote{false, in case $c_{i-1}$, $c_i$ and $c_{i+1}$ are collinear}  and we include this point in $E_i$.
In Figures~\ref{fig:1a},~\ref{fig:2a}, and \ref{fig:2b} the region $E_i$ is lightly shaded. 
By Lemma~\ref{lem:growing}, the $E_i$'s are disjoint.

Let $\sigma$ be a minimum link path from the initial to the final cop position.  Then $\sigma$ has at most $d$ bends.  Note that there are at most $n^2$ common tangents because there are at most $n^2$ pairs of points determining common tangents.  Our bounds derive from these two facts.

Our plan is to show that at each step we make progress in one of the following ways:
\begin{enumerate}
\item $c_i$ is a vertex \label{event:c-vertex}
\item $c_i$ and $b_i$ define a common tangent \label{event:cb-tangent}
\item $E_i$ contains a vertex of the region \label{event:vertex}
\item $E_i$ contains an endpoint of a common tangent with both endpoints in $R_i$ \label{event:common-tangent}
\item $E_i$ contains a bend of $\sigma$ \label{event:bend}  
\end{enumerate}

\remove{
Our plan is to show that $E_i$ contains either: (1) a vertex of the region; or (2) an endpoint of a common tangent \rednote{that was used to stop the cop's motion}; or (3) a bend of every minimum link path from the initial to final cop positions. 
There is one exception to this---when $c_i$ and $b_i$ define a common tangent---which we will take care of later on. 
}

We begin by bounding the number of events of each of the above types.  After that we will show that one of these events occurs in each step.

\remove{
We begin by bounding the number of events of types (1), (2), and (3).  After that, we will show that one of these events occurs in each step, and deal with the possible exception. 
}

\remove{
We  show that 
either $c_i$ and $b_i$ define a common tangent or else
$E_i$ contains either: (1) a vertex of the region; or (2) an endpoint of a common tangent; or (3) a bend of the minimum link path from the initial to final cop positions.   
There are at most $O(n^2)$ pairs of points defining common tangents.  
Next we bound the number of events involving $E_i$.
}

Because $c_i$ and $b_i$ never repeat, event~(\ref{event:c-vertex}) happens at most $n$ times and event~(\ref{event:cb-tangent}) happens at most $n^2$ times.  
Because the $E_i$'s are disjoint, event~(\ref{event:vertex}) happens at most $n$ times, event~(\ref{event:common-tangent}) happens at most $n^2$ times, and event~(\ref{event:bend}) happens at most $d$ times.  

\remove{
Event (1) can happen at most $n$ times because there are $n$ vertices.  
Event (3) can happen at most $d$ times since that is the maximum number of bends in any minimum link path between any two points.  
We claim that event (2) happens at most $O(n^2)$ times because there are $O(n^2)$ common tangents and a common tangent 
%whose terminus becomes part of 
with an endpoint in the exclusion region will never again be used to determine the cop's position. 
We will argue this for the previous cop move because that is easier to see in our figures.  In particular, we will show that if $c_i$ stopped at a common tangent $t$, then $t$ does not have an endpoint outside $R_i$ or on the boundary of $R_i$.  
Because $c_i$ stopped at common tangent $t$, thus $c_i \neq b_i$ and $t$ must cross $\ell_i$ at $c_i$.  
Since the boundary of $R_i$ is $\bar \ell_i$ which starts at $b_i$ and goes in the opposite direction but in the same line as $\ell_i$, therefore $t$ cannot have an endpoint outside $R_i$ or on its boundary.
}

% To justify this observe that the boundary of the active region $R_i$ is a straight segment from $b_i$ through $c_{i-1}$ and 
%the new cop position, $c_i$, is in the same line 
%and is outside the segment $b_i c_{i-1}$ unless $c_i = b_i$ in which case the cop stopped at $b_i$ because it is a vertex
%but not an interior point of the boundary.  Thus any common tangent that crosses the boundary of $R_i$ cannot go through $c_i$.  
%\bluenote{Well, unless $c_i = b_i$ and the common tangent goes through that point.  Which we can discount . . . }

It remains to prove that at each step, one of the above 5 events occurs.

\remove{
It remains to prove that at each step, one of event (1), (2), or (3) occurs, 
and to deal with the exception.
}

Recall the conditions for defining the cop's position $c_i$.  If $c_i=b_i$ because $b_i$ is a vertex then we have event~(\ref{event:c-vertex}). % because this vertex lies in $E_i$.
Otherwise
$c_i$ stops at a common tangent or a robber exit line or on the boundary of the region.  

Suppose first that $c_i$ stops on a common tangent, and not on the boundary (we will handle that case below).  Note that the common tangent must cross $\ell_i$, and therefore that both endpoints of the common tangent lie in $R_i$.   The boundary of $R_{i+1}$ is a line segment that goes through $c_i$ and therefore one endpoint of the common tangent must lie in $R_{i+1}$, possibly in the boundary.  This endpoint is in $E_i$.    Thus we have event~(\ref{event:common-tangent}).
\remove{If $c_i$ is on a common tangent,  then one endpoint of the common tangent must lie in the newly excluded region $E_i$, because the boundary of $E_i$ is a line segment going through $c_i$ and the common tangent line crosses this.}

If $c_i$ stops on a robber exit line then we claim that $E_i$ contains a vertex, i.e., event~(\ref{event:vertex}).  To justify this, first note that $r_{i-1}$ must lie in $R_{i+1}$.  This is because a straight segment joins $r_i$ and $r_{i-1}$ and $r_i$ lies on the far side of $\ell_{i+1}$, so in order for a straight  segment from $r_i$ to exit $R_{i+1}$ it would have to cross $\ell_{i+1}$ and $\bar \ell_{i+1}$ which is impossible.  
Therefore the robber exit line must exit $R_{i+1}$ at $c_i$ (or possibly lie in the line of $\ell_{i+1}$), and thus the tangent point of the robber exit line, and its bay, must lie in $E_i$.   As noted when we defined robber exit lines, this bay contains a vertex.  Therefore $E_i$ contains a vertex.    

Finally, we must consider the possibility that $c_i$ stops on the region boundary.   
When can this happen?
By Lemma~\ref{lem:infl}, the cop always stops at a common tangent in Case 2(a).  By Lemma~\ref{lem:exit}, Case 1(b) never occurs.  Thus we must be in Case 1(a) or 2(b).
If $c_i$ is not at a point where $\ell_i$ exits the region then $c_i$ and $b_i$ define a common tangent, i.e.~event~(\ref{event:cb-tangent}).
%Each of the $O(n^2)$ common tangents can occur at most once in this role, so this exceptional case can happen at most $O(n^2)$ times.} 
We are left with the case where $c_i$ is at a point where $\ell_i$ exits the region.
 In Case 2(b)  (see Figure~\ref{fig:2b}) the robber's move $r_{i-1} r_i$ must cross line $\ell_i$ beyond $c_i$, which is impossible if $\ell_i$ exits the region at $c_i$.  Thus we must be in Case 1(a). See Figure~\ref{fig:link-diameter}.  The minimum link path $\sigma$ from the initial cop position (outside $R_i$, or possibly on the boundary of $R_i$) to the final robber position (inside $R_{i+1}$) must include a bend point in $E_i$.   This is event~(\ref{event:bend}). 
\end{proof}

\begin{figure}[ht]
\centering
\includegraphics[width=.3\textwidth]{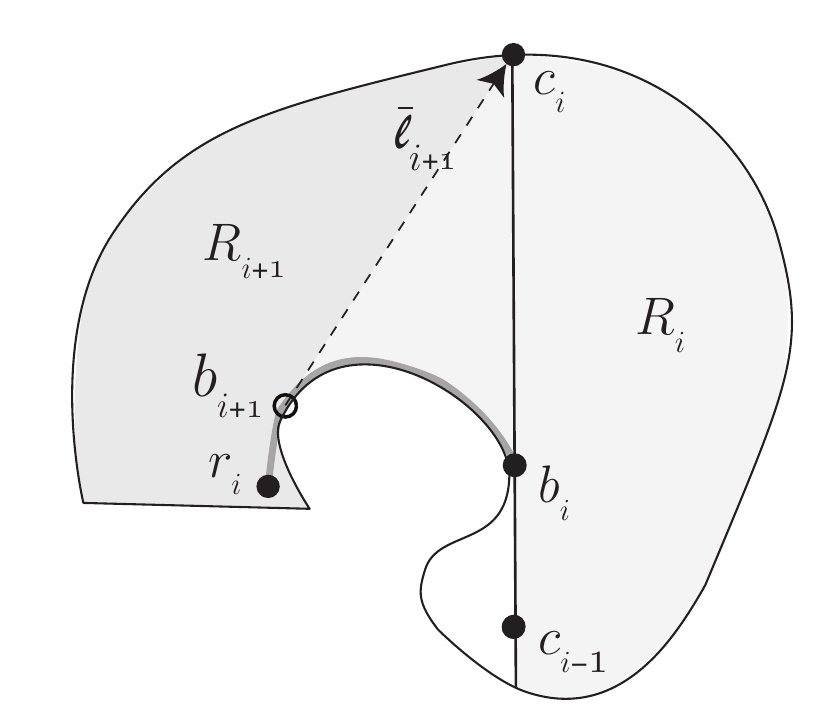}
\caption{When the cop stops on the region boundary in Case 1(a).}
\label{fig:link-diameter}
\end{figure}

%\rednote{This proof is not quite rigorous about corner cases.  To be more rigorous we could argue about $E_{i+1}$  . . . or we could be meticulous about collinear cases . . . }
%
\remove{
We consider different cases:
\begin{enumerate}
\item $E_i$ contains a vertex (curve endpoint). This happens every time $\ell_{i+1}$ has one or two endpoints on different curves than $\ell_i$.
\item $E_i$ contains a terminus of an inflection line. This happens every time the cop stops on an inflection line intersection.
The cop will never stop on the same inflection line, because it may not stop on an inflection line exiting $R_i$.
\item None of above conditions hold, i.e., the endpoints are on the same curves and cop has stopped because of an exit line or on the boundary.
We consider this situation in four cases we discussed before:
In case 2(a) always a curve endpoint is in $E_i$ so it would be filtered at item (1). Case 2(b) does not happen in this item, because we can show that either a curve endpoint is in $E_i$ or cop may not have stopped on robber's exit line. See Figure (?? \rednote{a figure to show it}). Case 1(b) never happens at all. Thus the only case we need to consider here is case 1(a) where the endpoints of $\ell_{i+1}$ are on the same curves as $\ell_i$. See Figure ??. Now, it is not possible to go from a point in $P - R_i$ to a point in $R_{i+1}$ without having a step in $E_i$, so the link distance from $c_0$ to active region is increased.
\end{enumerate}

\bluenote{Jack's version}
We bound the maximum number of cop moves by charging each move based on the new exclusion bay for the move after. The final two moves are charged against the cop catching the robber.

If the new exclusion bay lid has one or both endpoints on curves not in the previous exclusion bay, then the new exclusion bay has added the final portion of a curve to its endpoint, which can happen at most $2n$ times. I claim that shortest path cases 1 \bluenote{(2a)}
 and 2 \bluenote{(1a with bend at new curve (not $\gamma_i$))} always satisfy this condition. 
 For case 2, the tangency changes from $\gamma$ in a portion of curve that will become part of the exclusion bay. For case 1, the right turn in Figure 4.1 \bluenote{(Figure~\ref{fig:2a})}
  is either at an endpoint of a curve whose final portion is being added to the exclusion bay, or on a curve that must have an endpoint above the inflection line in the new exclusion bay.

This condition is also satisfied whenever a cop stops on an exit line---all path cases add the bay contributing the exit line to the exclusion bay. Just as argued for case 1, whether the exit line tangency is at an endpoint or curve, the last portion of a curve is added to the exclusion bay.

Consider when the cop stops on an inflection line and continues in case 3 \bluenote{(1a with bend still on $\gamma_i$)}
 or 4 \bluenote{(2b)}. 
 Neither endpoint of the inflection line could have been inside the previous exclusion bay, but because the new lid goes through $c$ to the boundary of $P$ in both directions, an inflection line endpoint is added to the exclusion bay. Since there are a quadratic number of inflection lines, this happens $O(n^2)$ times.

Finally, the cop can stop at the boundary and continue in case 3, so that the new portion of the exclusion bay (including the new lid and excluding the old) is bounded by two curves, curving the same way. Any straight-line path must have a vertex in this portion, so we charge this to the link diameter of $P$.
\rednote{Be more careful about link diameter.}
}

\remove{%%%%%%%%Hamide's version
We bound the maximum number of cop moves by charging each move based on the new exclusion region for the move after.  The final two moves are charged against the cop catching the robber.  

We first claim that the number of curves in $R_i$ is never more than $n+1$. That is because $\ell_i$ connects two points of the boundary and excludes all curves in between. In worst case it can possibly counted as a curve itself.
If $\ell_i$ has one or both endpoints on curves not in the previous exclusion region, then the new exclusion region has added the final portion of a curve to its endpoint, which can happen at most $2n$ times.  

We claim that if cop stops at $c_i$ because of an exit line intersection, this condition is always satisfied. 
That is because $r_{i-1}$ is visible to $c_i$ in this case. Either robber hides into a bay such that the shortest path from $c_i$ bends left or right, all the curves between the curve that form exit line and the curve that $b_i$ is on (that is non-empty) will be excluded. 

If cop stops at $c_i$ because of an inflection line intersection, then that inflection line will always have at least one endpoint in exclusion region, so it cannot be used again later. $\ell_i$ may not also add new active inflection line, because the only possible type of inflection line it can form is the tangents from new endpoints to curves which cannot be used later because will have an endpoint in exclusion region for all steps $j$, $j>i$.

Finally, the cop can stop at the polygon boundary and continue in case 3, so that the new portion of the exclusion bay (including the new lid and excluding the old) is bounded by two curves, curving the same way.  Any  straight-line path must have a vertex in this portion, so we charge this to the link diameter of~$R$.  \rednote{is it true}
\end{proof}
}

%\clearpage
%%%%%%%%%%%%%%%%%%%%%%%%%%%%%%%%%%%%%%%%%%%%%%%%
%%%%%%%%%%%%%%%%%%%%%%%%%%%%%%%%%%%%%%%%%%%%%%%%
%\section{Conclusions and Open Problems}
\section{Open Problems}

\begin{enumerate}
\item 
%\smallskip\noindent{\bf 1.}
Consider the cops and robbers game on the points inside 
%straight-line pursuit-evasion game in 
a polygonal region, i.e.,~a polygon with holes.  
There is a lower bound of three cops---an example requiring three cops can be constructed 
%\rednote{Give the example?} 
from a planar graph where three cops are required~\cite{Aigner-84} by taking a straight-line planar drawing of the graph and cutting out polygonal holes to leave narrow corridors for the graph edges.
It is an open question whether three cops suffice. 
%Do three cops suffice?  

\item 
\label{open:move-complexity}
%\smallskip\noindent{\bf 2.}
What is the complexity of finding how many moves the cop needs for a given polygon/region?  
%\mynote{Perhaps the Nowakowski-Winkler idea for inf. graphs just becomes a link distance thing in the polygon.}
The graph version of this problem is solvable in polynomial time for cop-win graphs~\cite{Hahn-MacGillivray-06}.
%, and more generally, for $k$-cop win graphs if $k$ is fixed~\cite{Bonato-capture-09}.  
For the cops and robbers game on the points inside a polygon we conjecture that the problem is solvable in polynomial time if the cop is restricted to the reflex vertices of the polygon.  
However, the cop may save by moving to an interior point---for example in a star-shaped polygon whose kernel is disjoint from the polygon boundary---so the problem seems harder if the cop is unrestricted.

\changed{
\item
Is there a lower bound of $\Omega(n^2 + d)$ on the worst case number of cop moves in a splinegon of $n$ curve segments and link diameter $d$?   From results in Section~\ref{sec:lower-bound} we have a lower bound of $\Omega(n+d)$. 
}

%\item What if the evader can move distance 2 in the link metric?  Does one cop suffice?  In the graph setting this problem has been considered by Frieze et al.\cite{Frieze-variations-12}.  They show that the cop number increases in general as the robber's speed increases.

\end{enumerate}

\section{Acknowledgements}

The cops and robbers problem for points inside a region with a curved boundary was initially suggested by Vinayak Pathak and was posed in the Open Problem Session of CCCG 2013 by the third author of this paper.  
We thank all who participated in discussing the problem, especially David Eppstein.

%%%%%%%%%%%%%%%%%%%%%%%%%%%%%%%%%%%%%%%%%%%%%%%%
%%%%%%%%%%%%%%%%%%%%%%%%%%%%%%%%%%%%%%%%%%%%%%%%

%\printbibliography

%\bibliographystyle{splncs}
%\small
\bibliographystyle{abbrv}
\bibliography{CR}

%%%%%%%%%%%%%%%%%%%%%%%%%%%%%%%%%%%%%%%%%%%%%%%%
%%%%%%%%%%%%%%%%%%%%%%%%%%%%%%%%%%%%%%%%%%%%%%%%

\clearpage
\begin{appendix}

%%%%%%%%%%%%%%%%%%%%%%%%%%%%%%%%%%%%%%%%%%%%%%%%
%%%%%%%%%%%%%%%%%%%%%%%%%%%%%%%%%%%%%%%%%%%%%%%%
 
\end{appendix}

\end{document}